\documentclass[11pt]{article}

\usepackage[latin1]{inputenc}
\usepackage[T1]{fontenc}
\usepackage[english]{babel}

\usepackage[top=20mm, bottom=20mm, left=20mm, right=20mm]{geometry}

\usepackage{graphicx}
\usepackage{float}

\usepackage{amsmath,amsthm,amssymb}
\usepackage{bbm}
\usepackage{mathtools}

\usepackage[math]{blindtext}

\usepackage{setspace}
\usepackage{natbib}
\usepackage{hyperref} 

\newtheorem{theorem}{Theorem}[section]

\newtheorem{lemma}[theorem]{Lemma}
\newtheorem{proposition}[theorem]{Proposition}
\newtheorem{formula}{Formula}
\newtheorem{assumption}{Assumption}

\newcommand{\res}{\mathrm{Res}}
\newcommand{\ud}{\mathrm{d}}
\newcommand{\K}{\mathrm{K}}
\DeclarePairedDelimiter{\ceil}{\lceil}{\rceil}

\begin{document}
\newgeometry{margin=2cm,top=2cm,bottom=2cm}

\title{
        Explicit option valuation in the exponential NIG model}

\author{Jean-Philippe Aguilar 
        \thanks{Cov\'ea Finance, 8 rue Boissy d'Anglas, FR-75008 Paris, 
        Email: jean-philippe.aguilar@covea-finance.fr}}
\date{October 04, 2020}

\maketitle
\thispagestyle{empty}

\begin{abstract}
We provide closed-form pricing formulas for a wide variety of path-independent options, in the exponential L\'evy model driven by the Normal inverse Gaussian process. The results are obtained in both the symmetric and asymmetric model, and take the form of simple and quickly convergent series, under some condition involving the log-forward moneyness and the maturity of instruments. Proofs are based on a factorized representation in the Mellin space for the price of an arbitrary path-independent payoff, and on tools from complex analysis. The validity of the results is assessed thanks to several comparisons with standard numerical methods (Fourier and Fast Fourier transforms, Monte-Carlo simulations) for realistic sets of parameters. Precise bounds for the convergence speed and the truncation error are also provided.

\bigskip


\noindent {\bfseries AMS subject classifications (MSC 2020):} 60E07, 60E10, 60H35, 65C30, 65T50, 91G20, 91G30.

\end{abstract}

\newpage
\setcounter{page}{1}

\section{Introduction}\label{sec:intro}

Whether the dramatic COVID-19 events and the subsequent turmoils in global markets were unpredictable "black swan" events in the sense of \cite{Taleb10} or, on the contrary, could have been forecasted (or at least, following the terminology of \cite{Giannone08}, "nowcasted") will undoubtedly be the matter of intense debates. But what is already certain is that they demonstrate, yet again, that the kurtosis in the distribution of asset returns far exceeds the tails of the Normal one, and that market volatility is not constant over time; it should therefore be a minimal requirement for any reliable  market model that they include (at least) these two stylized facts.

It has now long been known that exponential - sometimes also called geometrical - L\'evy models fulfil these conditions. Such models have been introduced in quantitative finance during the late 1990s / early 2000s in several influential works, and make the assumption that asset log returns are driven by some drifted L\'evy process: a Normal inverse Gaussian (NIG) process in \cite{Barndorff95,Barndorff97}, a Variance Gamma (VG) process in \cite{Madan98}, a hyperbolic or a generalized hyperbolic process in \cite{Eberlein95,Eberlein01}, a CGMY process in \cite{Carr02} or a stable (or $\alpha$-stable) process in \cite{Mittnik00,Carr03}. Stable distributions, in particular, may be noted for their historical importance, having been considered a credible candidate for the modelling of asset prices as early as in the 1960s by \cite{Mandelbrot63} in the context of the cotton market, thus paving the way to the more generic setup of exponential L\'evy models. Readers who may be less familiar with the broad family of L\'evy processes and their applications to finance are invited to refer to the classical references \cite{Bertoin96, Schoutens03, Cont04, Rachev11}.

In the present work, we will be particularly interested in the class of exponential L\'evy models whose L\'evy process is distributed according to a NIG distribution, namely, the class of exponential NIG models. NIG distributions were originally introduced for physical purpose, more precisely to model the complex behavior of dunes and beach sands, in the seminal article \cite{Barndorff77}; as noted above, they have subsequently been introduced for financial purpose approximately two decades later, because they feature several degrees of freedom that have a direct empirical interpretation in terms of financial time series. First, they possess fat tails, allowing for the presence of extreme variations of prices (positive or negative jumps); when the tail parameter goes to infinity, then the NIG distribution degenerates into the Normal distribution and the exponential NIG model recovers the Black-Scholes model (\cite{Black73}). Second, NIG distributions can be skewed, allowing to capture the asymmetry that can be observed in the distribution of jumps (price drops occurring more often than raises). Last, but not least, a NIG process can be interpreted as a drifted Brownian motion whose time follows an inverse Gamma process - this is a consequence of the fact that the NIG distribution is actually a particular case of a so-called Normal variance-mean mixture, the mixing distribution being the inverse Gaussian (IG) distribution; the NIG process is therefore a time changed L\'evy process, which allows for stochastic volatility modelling and related phenomena, such as clustering or negative correlation between the returns and their volatility (see details in \cite{Carr04}). Let us also mention that, as observed by \cite{Mechkov15} the NIG process is also deeply related to the Heston stochastic volatility model (\cite{Heston93}). Indeed, in the fast reversion limit, Heston log returns become NIG distributed, and the NIG parameters have a direct connection with the shape of the volatility surface. This allows for a simpler calibration to market data, and for the capture of a realistic smile in the short maturity region.

Of course, since it was introduced, the exponential NIG model has been proved to provide a very good fitting to financial data many times. Let us mention, among others, initial tests for daily returns on Danish and German markets in \cite{Barndorff95,Rydberg97} and subsequently on the FTSE All-share index (also known as "Actuaries index") in \cite{Venter02}. More recently, the impact of high frequency trading has also been taken into account, and calibrations have been performed on intraday returns e.g. in \cite{Figueroa12} for different sampling frequencies. Let us also mention that multivariate extensions of the exponential NIG model, i.e., featuring a different time change for different assets, have also been considered (see \cite{Luciano10} and references therein).

As one could expect, pricing contingent claims turns out to be a tougher task in the exponential NIG model than it is in the usual Black-Scholes framework. Numerical methods are largely favored, including Monte-Carlo valuation methods (\cite{Ribeiro03}), numerical evaluation of Fourier (\cite{Lewis01}) and Fast Fourier (\cite{Carr99}) transforms. The success of Fourier transform methods is strongly linked to the relative simplicity of the characteristic function of most exponential L\'evy models, and has opened the way to a wide range of other transform based approaches: they include, among others, the COS method by \cite{Fang08}, the Hilbert transform method (see notably a recent application to time-changed L\'evy processes in \cite{Zeng14}) or the local basis Frame PROJection (PROJ) method by \cite{Kirkby15}. Efforts have also been made towards analytic evaluation or approximations: in \cite{Ivanov13}, a closed-form formula (in terms of Appel functions) for the European call is derived in the particular case where the NIG distribution has a tail parameter of $1/2$, and in \cite{Albrecher04} approximations and bounds are provided for Asian options.

In this paper, we would like to show that it is actually possible to obtain tractable closed-form pricing formulas in the exponential NIG model, for a broad range of path independent instruments. This is made possible by a remarkable property allowing to express the Mellin transform of an arbitrary path independent option as the product of the Mellin transforms of its payoff and of the NIG probability density. Inverting it by means of residue summation yields the option price, computed under the form of quick convergent residue series whose terms are directly expressed in terms of the model's parameters. This Mellin residue summation method has been used very recently within the framework of other exponential L\'evy models, namely in the Finite Moment Log Stable (FMLS) model in \cite{AK19} and in the exponential VG model in \cite{Aguilar20}; in the present paper, we will therefore demonstrate that the technique is also well-suited to the exponential NIG model. Moreover, we will establish pricing formulas for both the symmetric and the asymmetric NIG processes, while the formulas in the VG case in \cite{Aguilar20} were mainly obtained for the symmetric VG process. Due to the nature of the residues series, however, we will need to introduce a restriction on the model parameters to ensure the convergence to the price. We will show that this condition is compliant with most of the implied parameters calibrated in the literature; moreover, when options are not far from the money, it is automatically satisfied.

The paper is organized as follows: in section \ref{sec:model}, we start by recalling fundamental concepts on the NIG process and its implementation via exponential L\'evy models. In section \ref{sec:symmetric}, we focus on the symmetric NIG process: after establishing the pricing formula in the Mellin space for an arbitrary path independent instrument, we evaluate, analytically, the price of the European and digital options, as well as payoffs featuring more exotic attributes (power options, log contracts, \dots). In section \ref{sec:asymmetric} we extend the pricing formula to the more general case of the asymmetric model, and provide analytic formulas for the digital and European prices. In section \ref{sec:num}, practical implementation is discussed, and precise bounds for the convergence speed and the truncation errors of the series are obtained; we also assess the validity of the results by comparing them with classic numerical methods (Fourier inversion, Monte Carlo simulations). For the reader's convenience, the paper is also equipped with two appendices: in appendix \ref{app:Mellin} we provide a short overview of the Mellin transform, and in appendix \ref{app:special} we recall some important special function identities that are used throughout the paper.

\section{Model definition}\label{sec:model}

In this section we recall important concepts on NIG distributions and processes; more details can be found in the initial articles by Barndorff-Nielsen or in subsequent review articles like \cite{Hanssen01,Papantoleon08}. We also introduce the exponential NIG model, following the classical setup of exponential L\'evy models such as defined e.g. in \cite{Schoutens03,Tankov10}.

\subsection{The Normal inverse Gaussian process}\label{subsec:NIG process}

The Normal inverse Gaussian (NIG) process can be defined by in several different ways: classically, it is defined either as a process whose increment follow a NIG distribution, in terms of its L\'evy measure, or as a time-changed L\'evy process. Let us also mention that, as remarked in \cite{Mechkov15}, the NIG process can also be seen as the limit of a Fast Reverting Heston (FRH) process.

\paragraph{NIG density}
The NIG distribution, denoted by $\mathrm{NIG}(\alpha,\beta,\delta,\mu)$, is a four-parameter distribution whose density function is:
\begin{equation}\label{NIG_distribution_density}
    f(x) \, := \, \frac{\alpha\delta}{\pi} \,
    e^{\delta\sqrt{\alpha^2 - \beta^2} +\beta (x-\mu)} \, 
    \frac{\K_1 \left( \alpha\sqrt{\delta^2 + (x-\mu)^2} \right) }{\sqrt{\delta^2 + (x-\mu)^2}}
    .
\end{equation}
The function $z\rightarrow \K_1(z)$ is the modified Bessel function of the second kind and of index 1 (sometimes also called Macdonald function, see definitions and properties in appendix \ref{app:special}). $\alpha > 0$ is a {\it tail} or {\it steepness} parameter controlling the kurtosis of the distribution; the large $\alpha$ regime gives birth to light tails, while small $\alpha$ corresponds to heavier tails. $\beta\in (-\alpha,\alpha-1)$ is the {\it skewness} parameter: $\beta < 0$ (resp. $\beta > 0$) implies that the distribution is skewed to the left (resp. the right), and $\beta =0$ that the distribution is symmetric around the {\it location} parameter $\mu\in\mathbb{R}$. $\delta > 0$ is the {\it scale} parameter and plays an analogue role to the variance term $\sigma^2$ in the Normal distribution; when $\beta=0$, the Normal distribution is itself recovered in the large steepness regime:
\begin{equation}
    \mathrm{NIG}(\alpha,0,\delta,\mu) \, \underset{\alpha \rightarrow \infty}{\longrightarrow} \, \mathcal{N} (\mu,\sigma^2) \, ,  \hspace{0.3cm} \sigma^2:= \frac{\delta}{\alpha} 
    .
\end{equation}

We say that a stochastic process $\{ X_t \}_{t\geq 0}$ is a NIG process if it has NIG distributed increments, that is if $ X_{t+h} \, - \, X_t \,  \sim \, \mathrm{NIG}(\alpha,\beta,\delta h,\mu h) $ for all $h\geq 0$; it follows from \eqref{NIG_distribution_density}  that the density of the process conditionally to $X_0=0$ is (with a slight abuse of notations):
\begin{equation}\label{NIG_process_density}
    f(x,t) \, := \, \frac{\alpha\delta t}{\pi} \,
    e^{\delta t \sqrt{\alpha^2 - \beta^2} +\beta (x-\mu t )} \, 
    \frac{\K_1 \left( \alpha\sqrt{(\delta t)^2 + (x-\mu t)^2} \right) }{\sqrt{(\delta t)^2 + (x-\mu t )^2}}
    .
\end{equation}
It is also possible to define the NIG process as a time-changed drifted Brownian motion: if $\{I_t\}_{t\geq 0}$ is a process distributed according to an Inverse Gamma density of shape $\delta\sqrt{\alpha^2 - \beta^2}$ and mean rate 1 and if $\{ W_t \}_{t \geq 0}$ is a standard Wiener process, then the process
\begin{equation}\label{NIG_subordination}
    X_t \, = \, \beta \delta^2 \, I_t \, + \, \delta W_{I_t}
\end{equation}
is a centered NIG process ($\mu = 0$). The process $\{I_t\}_{t\geq 0}$ is a tempered stable subordinator; it has positive jumps, and therefore is interpreted as a {\it business time} that can differ from the {\it operational time}, the occurence of jumps corresponding to periods of intense business activity. A similar interpretation holds for instance in the case of the Variance Gamma process, which features another example of tempered stable subordination (via a Gamma process). 

\paragraph{L\'evy symbol}
The NIG process is a (pure jump) L\'evy process whose characteristic function $\Psi(u,t):=\mathbb{E} [ e^{i u X_t} ] $ can be written down as $\Psi(u,t) \, = \, e^{ t  \psi(u)}$, where the characteristic exponent, or L\'evy symbol, is known in exact form:
\begin{equation}\label{NIG_characteristic}
    \psi (u) \
     := \, \log \Psi (u,1) 
    \, = \, i \mu u \,- \, \delta 
    \left( \sqrt{ \alpha^2 - (\beta + i u)^2 } - \sqrt{ \alpha^2 - \beta^2 }
    \right).
\end{equation}
The process admits the L\'evy-Khintchine triplet $ (a,0,\nu (\ud x) )$, where the drift $a$ and the L\'evy measure $\nu$ are defined by
\begin{equation}
    \left\{
    \begin{aligned}
        & a \, := \mu \, + \, \frac{2\alpha \delta}{\pi} \, 
        \int\limits_0^1 \, \mathrm{sinh}(\beta x) \K_1(\alpha x) \, \ud x
        \\
        & \nu (\ud x) \, :=  \frac{\alpha\delta}{\pi} \, e^{\beta x} \, \frac{\K_1(\alpha |x| )}{|x|} \, \ud x
        ,
    \end{aligned}
    \right.
\end{equation}
allowing to write down the characteristic exponent \eqref{NIG_characteristic} in terms of its L\'evy-Khintchine representation:
\begin{equation}
    \psi (u) \,  = \, 
    i a u \, + \, 
    \int\limits_{\mathbb{R}} \, ( e^{i u x} - 1 - i u x \mathbbm{1}_{ \{ |x| < 1  \} } ) \, \nu (\ud x)  
    .
\end{equation}
Let us observe that it follows from the definition of the L\'evy measure $\nu$ that the NIG process has infinite variation and infinite intensity (i.e. $\nu(\mathbb{R}) = \infty$), and therefore possesses a very rich dynamics with infinite number of jumps on any time interval - this is why no Brownian component is even needed in the L\'evy-Khintchine triplet. We should also note that the NIG process has all its moments finite, which is not the case with (double-sided) $\alpha$-stable processes for instance: this is because the Bessel function admits the asymptotic behavior (see \eqref{Bessel_large_z})
\begin{equation}
    \K_1 ( |x| ) \, \underset{|x|\rightarrow\infty}{\sim} \, \sqrt{\frac{\pi}{2 |x|}} \, e^{-|x|},
\end{equation}
and therefore the tails of the NIG measure $\nu$ are less heavy
than the tails of the $\alpha$-stable measure (which has polynomial decrease in $ 1/ |x|^{1+\alpha})$. In other words, the jumps in the NIG process are not as big as for $\alpha$-stable processes, but allow finiteness of moments and therefore of option prices; in the $\alpha$-stable case, this would be achieved only for spectrally negative processes (i.e., having negative jumps only).

\paragraph{Fast reverting Heston limit}
Following \cite{Mechkov15}, we choose the following formulation for the Heston dynamics:
\begin{equation}
    \left\{
    \begin{aligned}
        & \ud x_t \, = \, -\frac{1}{2} \sigma^2 z_t \ud t \, + \, \sigma\sqrt{z_t}\ud W_t^{(1)} \\
        & \ud z_t \, = \, a \left( (1-z_t)\ud t \, + \, \gamma\sqrt{z_t} \ud W_t^{(2)}     \right)
    \end{aligned}
    \right.
\end{equation}
where the two Brownian motions have correlation $\rho$, and we consider the fast reversion limit $a\rightarrow\infty$ in the CIR process driving the stochastic multiplier $z_t$. Then, in this limit, the process $\left\{x_t\right\}_{t\geq 0}$ is distributed according to a NIG distribution:
\begin{equation}
    x_t \, \sim \, \mathrm{NIG} \, \left( \frac{\sqrt{4-4\rho\gamma\sigma+\gamma^2\sigma^2}}{2\gamma\sigma(1-\rho^2)} , 
    - \frac{\gamma\sigma - 2\rho}{2\gamma\sigma(1-\rho^2)}  ,
    \frac{\sigma\sqrt{1-\rho^2}}{\gamma}t , 
    -\frac{\sigma\rho}{\gamma} t
    \right)
    .
\end{equation}
Let us mention that the NIG process is also associated to the long maturity asymptotic of the original (non fast reverting) Heston model (see e.g. \cite{Keller08}). 

\subsection{The exponential NIG model}\label{subsec:exp NIG model}

\paragraph{Model specification}
Let $T>0$ and $S: t \in [0,T]\rightarrow S_t$ be the market price of some financial asset, seen as the realization of a time dependent random variable $\{S_t\}_{t\in[0,T]}$ on the canonical space $\Omega=\mathbb{R}_+$ equipped with its natural filtration. We assume that there exists a risk-neutral measure $\mathbb{Q}$ under which the instantaneous variations of $S_t$ can be written down as:
\begin{equation}\label{Exp_NIG_SDE}
    \frac{\ud S_t}{S_t} \, = \, (r-q) \, \ud t \, + \, \ud X_t
\end{equation}
where $r\geq 0$ is the risk-free interest rate and $q\geq 0$ is the dividend yield (both assumed to be deterministic and continuously compounded), and where $\{ X_t \}_{t\in[0,T]}$ is the NIG process. The solution to the stochastic differential equation $\eqref{Exp_NIG_SDE}$ it the exponential process
\begin{equation}\label{Exp_process}
    S_T \, = \, S_t \, e ^{ (r - q + \omega)\tau \, + \, X_\tau} , \hspace{0.3cm} \omega:= - \psi ( -i )
    ,
\end{equation}
where $\tau:= T - t$ is the time horizon and $\omega$ is the martingale adjustment (also called convexity adjustment, or compensator) determined by the martingale condition $\mathbb{E}^{\mathbb{Q}} [ S_T \vert S_t ] = e^{(r-q)\tau} S_t$; it follows from the definition of the L\'evy symbol \eqref{NIG_characteristic} that this adjustment is equal to:
\begin{equation}\label{omega}
    \omega \, = \,  - \mu  \, + \, \delta 
    \left( \sqrt{ \alpha^2 - (\beta + 1 )^2 } - \sqrt{ \alpha^2 - \beta^2 }
    \right)
    .
\end{equation}
It is interesting to note that, in the large steepness regime, \eqref{omega} has the following asymptotic behavior:
\begin{equation}\label{omega_large_alpha}
    \omega \, \underset{\alpha\rightarrow\infty}{\sim} \, - \mu \, - \, \frac{\sigma^2}{2} (1+2\beta)  \, ,  \hspace{0.3cm} \sigma^2:= \frac{\delta}{\alpha} 
    .
\end{equation}
Taking $\mu=0$ (centered process) and $\beta=0$ (symmetric process), \eqref{omega_large_alpha} recovers the the Gaussian martingale adjustment $-\sigma^2 / 2$, and the exponential NIG model \eqref{Exp_NIG_SDE} degenerates into the Black-Scholes model.

\paragraph{Contingent claim valuation}
Given a path-independent payoff function $\mathcal{P}$, i.e., a positive function depending only on the terminal value $S_T$ of the market price and on some strike parameters $K_1, \dots , K_N \, > \, 0$, then the value at time $t$ of a contingent claim delivering a payoff $\mathcal{P}$ at maturity is equal to the following risk-neutral expectation:
\begin{equation}\label{Risk-neutral_1}
   \mathcal{C} \, = \, \mathbb{E}^{\mathbb{Q}} \left[e^{-r \tau} \mathcal{P} (S_T, K_1, \dots , K_n  ) \, \vert \, S_t \right]
   .
\end{equation}
The conditional expectation \eqref{Risk-neutral_1} can be achieved by integrating all possible realizations for the payoff over the probability density of the NIG process, thus resulting in:
\begin{equation}\label{Risk-neutral_2}
     \mathcal{C} \, = \, e^{-r\tau} \, \int\limits_{-\infty}^{+\infty} \, \mathcal{P} (S_t \, e^{(r-q+\omega)\tau \, + \, x},K_1, \dots , K_n) \, f(x,\tau) \, \ud x
     .
\end{equation}

\section{Option pricing in the symmetric model}\label{sec:symmetric}

In this section, we assume that $\beta=0$, i.e., that the process $\{ X_t \}_{t\in[0,T]}$ in \eqref{Exp_NIG_SDE} is distributed according to the symmetric distribution $\mathrm{NIG}(\alpha, 0 , \delta t, \mu t)$. First, we establish a general pricing formula for an arbitrary path independent instrument; then, we apply this formula to the analytic evaluation of several options and contracts.

\subsection{Pricing formula}\label{subsec:sym_pricing}

Let us start by establishing a representation for the symmetric NIG density $f(x,t)$ under the form of a Mellin-Barnes integal.

\begin{lemma}\label{lemma:sym_density}
    For any $c_1 \in \mathbb{R}_+$, the following holds true:
    \begin{equation}\label{sym_density}
        f(x,t) \, = \, \frac{\alpha}{2\pi} \, e^{\alpha \delta t} \,
        \int\limits_{c_1 - i\infty}^{c_1 + i\infty} \,
        \Gamma\left( \frac{s_1}{2} \right) \,
        \K_{1-\frac{s_1}{2}}(\alpha\delta t) \,
        \left( \frac{2\delta t}{\alpha} \right)^{\frac{s_1}{2}} \,
        |x - \mu t|^{-s_1} \,
        \frac{\ud s_1}{2i\pi}
        .
    \end{equation}
\end{lemma}
\begin{proof}
    Taking $\beta=0$ in \eqref{NIG_process_density} yields:
    \begin{equation}\label{sym_density_1}
    f(x,t) \, = \, \frac{\alpha\delta t}{\pi} \,
    e^{\alpha \delta t } \, 
    \frac{\K_1 \left( \alpha\sqrt{(\delta t)^2 + (x-\mu t)^2} \right) }{\sqrt{(\delta t)^2 + (x-\mu t )^2}}
    .
    \end{equation}
    Using the Mellin transform for the Bessel function (see table \ref{tab:Mellin} in appendix \ref{app:Mellin} with $\nu=1$) and the Mellin inversion formula \eqref{inversion}, we can write:
    \begin{equation}\label{BesselK_Mellin}
        \frac{\K_1 \left( \alpha\sqrt{(\delta t)^2 + (x-\mu t)^2} \right) }{\sqrt{(\delta t)^2 + (x-\mu t )^2}} 
        \, = \, 
        \frac{1}{2\delta\tau} \, 
        \int\limits_{c_1 - i \infty}^{c_1+ i \infty} \, 
        \Gamma\left(\frac{s_1}{2}\right) \, 
        \K_{1-\frac{s_1}{2}}(\alpha\delta\tau) \, 
         \left( \frac{2\delta t}{\alpha} \right)^{\frac{s_1}{2}} \, 
         |x-\mu t|^{-s_1} \, 
         \frac{\ud s_1}{2 i \pi}
    \end{equation}
    for any $c_1>0$. Inserting into \eqref{sym_density_1} yields the representation \eqref{sym_density}.
\end{proof}

Let us now introduce the double-sided Mellin transform of the payoff function:
\begin{equation}\label{sym_payoff_transform}
    P^*(s_1) \, = \, \int\limits_{-\infty}^{\infty} \, \mathcal{P} \left( S_t e^{(r - q+\omega)\tau +x} , K_1, \dots , K_n \right) \, |x - \mu\tau|^{-s_1} \, \ud x
\end{equation}
and assume that it exists for $Re(s_1) \in (c_-,c_+)$ for some real numbers $c_- < c_+$. Then, as a consequence of the risk-neutral pricing formula \eqref{Risk-neutral_2} and of lemma \ref{lemma:sym_density}, we immediately obtain:

\begin{proposition}[Factorization in the Mellin space]
    \label{prop:symm_factorization}
    Let $c_1\in (\tilde{c}_-,\tilde{c}_+)$ where $
    (\tilde{c}_-,\tilde{c}_+) := (c_-,c_+) \cap \mathbb{R}_+$ is assumed to be nonempty. Then the value at time $t$ of a contingent claim delivering a payoff $\mathcal{P}(S_T,K_1, \dots , K_n)$ at its maturity $t=T$ is equal to:
    \begin{equation}\label{sym_factorization}
        \mathcal{C}  \, = \, 
        \frac{\alpha}{2\pi} \, e^{(\alpha \delta - r) \tau} \, \int\limits_{c_1-i\infty}^{c_1+i\infty} \, 
        \Gamma\left( \frac{s_1}{2} \right) \,
        P^*(s_1) \,
        \K_{1-\frac{s_1}{2}}(\alpha\delta \tau) \,
        \left( \frac{2\delta \tau}{\alpha} \right)^{\frac{s_1}{2}}  \, \frac{\ud s_1}{2i\pi}
        .
    \end{equation}
\end{proposition}
Throughout the paper, our purpose will be to express the complex integral \eqref{sym_factorization} as a sum of residues associated to the singularities of the integrand. Schematically, we will therefore be able to express the price of a contingent claim under the form of a series:
\begin{equation}
    \frac{\alpha}{2\pi} \, e^{(\alpha \delta - r) \tau} \, \times \, \sum \,
    \left[
    \textrm{residues of }  \Gamma\left( \frac{s_1}{2} \right) \,
     P^*(s_1)  \, 
    \times
    \textrm{particular values of } \K_{1-\frac{s_1}{2}}(\alpha\delta \tau) \,
    \times
    \textrm{powers of } \frac{2\delta\tau}{\alpha}
    \right]
    .
\end{equation}
As we will see, the residues turn out to have to be computed in the multidimensional sense, because, depending on the payoff's complexity, the evaluation of $P^*(s_1)$ can call for the introduction of a second Mellin variable $s_2$ (in the asymmetric case, we will see that one even needs a third Mellin variable $s_3$). However, as only Gamma functions are involved, these residues are straightforward to compute, even in the $\mathbb{C}^n$ sense.

Before proceeding to pricing itself, let us introduce the notation for the forward strike $F$ and the log forward moneyness $k$:
\begin{equation}\label{moneyness}
   F \, := \, Ke^{-(r-q)\tau} ,
   \hspace*{1cm}
   k \, := \, \log\frac{S_t}{F} \, + \, \omega\tau \, = \, \log\frac{S_t}{K} + (r - q + \omega) \tau
   .
\end{equation}
It will also be useful to introduce $k_0:= k + \mu\tau$; taking $\beta=0$ in the definition of the martingale adjustment \eqref{omega}, we have:
\begin{equation}
    k_0 \, = \, \log\frac{S_t}{K} \, + \, 
    \left(r - q + \delta \left( \sqrt{\alpha^2-1} - \alpha \right) \right) \tau 
    .
\end{equation}
Note that $k_0$ is independent of the location $\mu$ (in both the symmetric and asymmetric cases). Last, we need to introduce a restriction on the parameters, that will be fundamental for the series to converge:
\begin{assumption}\label{ass:growth}
    In all of the following, and unless otherwise stated, we will assume that the model's inputs are such that
    \begin{equation}
        \frac{ \vert k_0 \vert}{\delta\tau}  \, < \, 1
        .
    \end{equation}
\end{assumption}

\subsection{Digital and European options}\label{subsec:sym_dig_eur}

We start our applications of proposition \ref{prop:symm_factorization} with the determination of the price of the digital (also called binary) options, and of the vanilla European option.

\paragraph{Digital option (asset-or-nothing)} The asset-or-nothing call option consists in receiving a unit of the underlying asset $S_T$, on the condition that it exceeds a predetermined strike price $K$. The payoff can therefore be written down as:
\begin{equation}
    \mathcal{P}_{a/n} (S_T,K) \, := \, S_T \, \mathbbm{1}_{ \{ S_T  > K \}  }
    .
\end{equation}

\begin{formula}[Asset-or-nothing call] 
    \label{formula:sym_A/N}
    The value at time $t$ of an asset-or-nothing call option is: \\
    \begin{equation}\label{sym_A/N}
        C_{a/n} \, = \, \frac{ K\alpha e^{(\alpha \delta - r) \tau}}{\sqrt{\pi}} \, 
        \sum\limits_{\substack{n_1 = 0 \\ n_2 = 0}}^{\infty} \,
        \frac{ k_0^{n_1} }{n_1! \Gamma( 1+\frac{-n_1+n_2}{2} )} \,
        \K_{\frac{n_1-n_2+1}{2}}(\alpha\delta\tau)
        \left(  \frac{\delta\tau}{2\alpha} \right)^{\frac{-n_1+n_2+1}{2}}
        .
    \end{equation} 
\end{formula}
\begin{proof}
    \underline{Step 1:} Let us first assume that $k_0 < 0$. We remark that, using notations \eqref{moneyness}, we can write
    \begin{equation}
        \mathcal{P}_{a/n} \left( S_t e^{(r-q+\omega)\tau +x} , K \right)
        \, = \,
        K \, e^{k+x} \, \mathbbm{1}_{ \{ x > - k \} }
        .
    \end{equation}
     Using a Mellin-Barnes representation for the exponential term (see table \ref{tab:Mellin} in appendix \ref{app:Mellin}):
    \begin{equation}\label{MB_payoff_a/n}
        e^{k +x} \, = \, \int\limits_{c_2 - i\infty}^{c_2 + i\infty} (-1)^{-s_
        2} \Gamma(s_2) (k  + x)^{-s_2} \, \frac{\ud s_2}{2i\pi} \hspace{1cm} (c_2 > 0)
    \end{equation}
    and inserting into \eqref{sym_payoff_transform}, we get:
    \begin{align}\label{sym_A/N_payoff_Mellin}
        P^*(s_1) & = \, K \int\limits_{c_2 - i\infty}^{c_2 + i\infty} (-1)^{-s_2} \Gamma(s_2) 
        \int\limits_{ -k }^\infty (k + x)^{-s_2} (x - \mu\tau)^{-s_1} \, \ud x \, \frac{\ud s_2}{2i\pi} \\
        & = K \int\limits_{c_2 - i\infty}^{c_2 + i\infty} (-1)^{-s_2} \frac{\Gamma(s_2)\Gamma(1-s_2)\Gamma(s_1+s_2-1)}{\Gamma(s_1)} (-k_0)^{-s_1-s_2+1}   
        \, \frac{\ud s_2}{2i\pi}
    \end{align}
    where the $x$-integral exists because $-(k+\mu\tau) = -k_0 >0$ by hypothesis. Using proposition \ref{prop:symm_factorization} and the Legendre duplication formula \eqref{Legendre}, we obtain the price of the asset-or-nothing call:
    \begin{multline}\label{sym_A/N_int}
        \mathcal{C}_{a/n} \, = \, \frac{ K\alpha e^{(\alpha \delta - r) \tau}}{\sqrt{\pi}} \, 
        \int\limits_{c_1 - i\infty}^{c_1+i\infty}
        \int\limits_{c_2 - i\infty}^{c_2+i\infty}
        (-1)^{-s_2} \frac{\Gamma(s_2)\Gamma(1-s_2)\Gamma(s_1+s_2-1)}{\Gamma(\frac{s_1+1}{2})} (-k_0)^{-s_1-s_2+1}\K_{1-\frac{s_1}{2}}(\alpha\delta\tau)
        \\
        \times \left( \frac{\delta\tau}{2\alpha}  \right)^{\frac{s_1}{2}} \, \frac{\ud s_1}{2i\pi} \frac{\ud s_2}{2i\pi}
    \end{multline}
    which converges in the subset $\{ (s_1,s_2)\in\mathbb{C}^2, 0 < Re(s_2) < 1, Re(s_1+s_2) > 1  \}$ and can be analytically continued outside this polyhedron, except when the Gamma functions in the numerator are singular, that is, when their arguments equal a negative integer. If we consider the singularities induced by $\Gamma(s_2)$ at $s_2 = -n_2$, $n_2\in\mathbb{N}$ and by $\Gamma(s_1+s_2-1)$ at $s_1+s_2-1 = -n_1$, $n_1\in\mathbb{N}$, then, the associated residues are straightforward to compute via the change of variables $u:=s_1+s_2-1$, $v:=s_2$, and via the singular behavior \eqref{sing_Gamma} for the Gamma functions; they read:
    \begin{equation}\label{sym_A/N_res}
        \frac{ K\alpha e^{(\alpha \delta - r) \tau}}{\sqrt{\pi}}
        (-1)^{n_2}
        \frac{(-1)^{n_1}}{n_1!} \frac{(-1)^{n_2}}{n_2!}
        \frac{\Gamma(1+n_2)}{\Gamma(1+\frac{-n_1+n_2}{2})}
        (-k_0)^{n_1}
        \K_{\frac{n_1-n_2+1}{2}}(\alpha\delta\tau)
        \left( \frac{\delta\tau}{2\alpha} \right)^{\frac{-n_1+n_2+1}{2}}
        .
    \end{equation}
    Simplifying and summing all residues \eqref{sym_A/N_res} yields the announced series \eqref{sym_A/N}.
    
    \noindent \underline{Step 2:} Let us now assume that $k_0>0$: in that case, the $x$-integral on the interval $(-k_0,\infty)$ in \eqref{sym_A/N_payoff_Mellin} does not converge. But, as $\{ X_t \}_{t\in[0,T]}$ is a $\mathbb{Q}$-martingale, we can write:
    \begin{equation}
        \mathbb{E}^\mathbb{Q} [ S_T \, \mathbbm{1}_{ \{ S_T > K \} }  \, \vert \, S_t ]
        \, = \, 
        S_t \, e^{(r-q)\tau} \, - \, 
        \mathbb{E}^\mathbb{Q} [ S_T \, \mathbbm{1}_{ \{ S_T < K \} }  \, \vert \, S_t ]
        .
    \end{equation}
    To compute the expectation in the r.h.s., we apply exactly the same technique than in step 1 (in this case, the $P^*(s_1)$ function exists, as an integral over $(-\infty , -k_0)$), resulting in the same residue formula than \eqref{sym_A/N_res}.
    
    \noindent \underline{Step 3:} Last, we have to examine the convergence of the series; to that extent let us denote the general term of the series \eqref{sym_A/N} by:
    \begin{equation}
        R_{n_1,n_2} \, := \,   \frac{ k_0^{n_1} }{n_1! \Gamma( 1+\frac{-n_1+n_2}{2} )} \,
        \K_{\frac{n_1-n_2+1}{2}}(\alpha\delta\tau)
        \left(  \frac{\delta\tau}{2\alpha} \right)^{\frac{-n_1+n_2+1}{2}}    .
    \end{equation}
    Let us fix $n_2\in\mathbb{N}$ and let $n_1\rightarrow\infty$; without loss of generality and to simplify the notations we can assume e.g. $n_2=0$ and study the behavior of
    \begin{equation}
        R_{n_1} \, := \, \frac{ k_0^{n_1} }{n_1! \Gamma( 1 - \frac{n_1}{2} )} \,
        \K_{\frac{n_1+1}{2}}(\alpha\delta\tau)
        \left(  \frac{\delta\tau}{2\alpha} \right)^{\frac{-n_1+1}{2}}  
        .
    \end{equation}  
    We may note also that, due to the presence of the $\Gamma(1-\frac{n_1}{2})$ function in the denominator, only odd terms $n_1 = 2p+1$ survive when $n_1\geq 1$. Using the particular value of the Gamma function \eqref{gamma_half_integers}, we are left with:
    \begin{equation}\label{R_p}
        R_{2 p +1} \, = \,  \frac{1}{\sqrt{\pi}}\frac{1}{2p+1}\frac{(-1)^p}{4^p p!} \, k_0^{2p+1}\K_{p+1}(\alpha\delta\tau)\left(  \frac{\delta\tau}{2\alpha} \right)^{-p}  
        .
    \end{equation}
    Using the Stirling approximation \eqref{Stirling} for $p!$ and the large index behavior \eqref{Bessel_large_index} for $\K_{p+1}(\alpha\delta\tau)$ and simplifying, we get:
    \begin{equation}\label{R_p_infty}
        |R_{2p+1}| \, \underset{p\rightarrow\infty}{\sim} \, \frac{1}{\sqrt{2\pi p (2p+2)}} \, \frac{k_0}{e\alpha\delta\tau} \, \left( \frac{k_0^2}{(\delta\tau)^2} \right)^p
    \end{equation}
    and therefore the series converge if and only if $\frac{k_0^2}{(\delta\tau)^2} < 1$, which is equivalent to assumption \ref{ass:growth}. Last, if we fix $n_1$,  then the symmetry relation \eqref{Bessel_sym} for the modified Bessel function and similar arguments (special values of the Gamma function and Stirling approximation) show that the series converge for all parameter values when $n_2\rightarrow\infty$.
\end{proof}

\paragraph{European option} The European call pays $S_T-K$ at maturity, at the condition that  the spot price is greater that the strike price. The payoff can therefore be written down as:
\begin{equation}
    \mathcal{P}_{eur} (S_T,K) \, := \, [S_T \, - \, K]^+
    .
\end{equation}

\begin{formula}[European call] 
    \label{formula:sym_eur}
    The value at time $t$ of a European call option is: \\
    \begin{equation}\label{sym_eur}
        C_{eur} \, = \, \frac{ K\alpha e^{(\alpha \delta - r) \tau}}{\sqrt{\pi}} \, 
        \sum\limits_{\substack{n_1 = 0 \\ n_2 = 1}}^{\infty} \,
        \frac{ k_0^{n_1} }{n_1! \Gamma( 1+\frac{-n_1+n_2}{2} )} \,
        \K_{\frac{n_1-n_2+1}{2}}(\alpha\delta\tau)
        \left(  \frac{\delta\tau}{2\alpha} \right)^{\frac{-n_1+n_2+1}{2}}
        .
    \end{equation} 
\end{formula}
\begin{proof}
We remark that, using notations \eqref{moneyness}, we can write:
    \begin{equation}
        \mathcal{P}_{eur} (Se^{(r - q +\omega)\tau+x},K) 
        \, = \,
        K (e^{k + x} -1) \mathbbm{1}_{ \{x > -k \} }
        .
    \end{equation}
    Then, we use the Mellin-Barnes representation (see table \ref{tab:Mellin} in appendix \ref{app:Mellin}):
    \begin{equation}\label{MB_payoff_european}
        e^{k + x} - 1 \, = \, \int\limits_{c_2 - i\infty}^{c_2 + i\infty} (-1)^{-s_
        2} \Gamma(s_2) (k  + x)^{-s_2} \, \frac{\ud s_2}{2i\pi} \hspace{1cm} (-1 < c_2 < 0)
    \end{equation}
    and we proceed exactly the same way than for proving Formula \ref{formula:sym_A/N}; note that the $n_2$-summation in \eqref{sym_eur} now starts in $n_2=1$ instead of $n_2=0$, because the strip of convergence of \eqref{MB_payoff_european} is reduced to $<-1,0>$ instead of $<0,\infty>$ in \eqref{MB_payoff_a/n}.
\end{proof}
Let us examine the series \eqref{sym_eur} in the large steepness regime ($\alpha\rightarrow\infty$). It follows from the asymptotic behavior of the Bessel function for large arguments \eqref{Bessel_large_z} that:
\begin{equation}
    \K_{\frac{n_1-n_2+1}{2}}(\alpha\delta\tau) \, \underset{\alpha\rightarrow\infty}{\sim} \, \frac{\sqrt{\pi}}{\sqrt{2\alpha\delta\tau}} \, e^{-\alpha\delta\tau} 
    ,
\end{equation}
and from \eqref{omega_large_alpha} that:
\begin{equation}
    k_0 \, \underset{\alpha\rightarrow\infty}{\sim} \, \log\frac{S_t}{K} \, + \, \left(r - q -\frac{\delta}{2\alpha} \right) \tau
    .
\end{equation}
Therefore, denoting $\sigma^2 := \frac{\delta}{\alpha}$, we obtain
\begin{equation}\label{sym_eur_BS}
    C_{eur}^{(\alpha\rightarrow\infty)} \, =  \, 
    \frac{ K e^{ - r \tau}}{2} \, 
    \sum\limits_{\substack{n_1 = 0 \\ n_2 = 1}}^{\infty} \,
    \frac{ 1 }{n_1! \Gamma( 1+\frac{-n_1+n_2}{2} )} \,
    \left( \log\frac{S_t}{K} + \left(r-q-\frac{\sigma^2}{2} \right) \tau  \right)^{n_1}
    \left(  \frac{\sigma^2\tau}{2} \right)^{\frac{-n_1+n_2}{2}}
\end{equation}
which is the series expansion of the Black-Scholes formula for the European call that was derived in \cite{Aguilar19}.

\paragraph{Digital option (cash-or-nothing)}

The payoff of the cash-or-nothing call option is
\begin{equation}
    \mathcal{P}_{c/n}(S_T,K) \, = \, \mathbbm{1}_{ \{ S_T>K \} }
\end{equation}
and therefore the option price itself is:
\begin{equation}
    C_{c/n} \, = \, \frac{1}{K} \, \left(  C_{a/n} - C_{eur})  \right)
    .
\end{equation}
Using formulas~\ref{formula:sym_A/N} and \ref{formula:sym_eur}, it is immediate to see that:

\begin{formula}[Cash-or-nothing call] 
    \label{formula:sym_c/n}
    The value at time $t$ of a cash-or-nothing call option is: \\
    \begin{equation}\label{sym_c/n}
        C_{c/n} \, = \, \frac{  \alpha e^{(\alpha \delta - r) \tau}}{\sqrt{\pi}} \, 
        \sum\limits_{n = 0}^{\infty} \,
        \frac{ k_0^{n} }{n! \Gamma( 1-\frac{n}{2} )} \,
        \K_{\frac{n+1}{2}}(\alpha\delta\tau)
        \left(  \frac{\delta\tau}{2\alpha} \right)^{\frac{-n+1}{2}}
        .
    \end{equation} 
\end{formula}

In \eqref{sym_c/n}, only terms for $n=0$ and $n=2p+1$, $p\in\mathbb{N}$ actually survive (because of the divergence of the Gamma function in the denominator when $n=2p$, $p\geq 1$). Therefore, using the particular values of the Gamma function at negative half-integers \eqref{gamma_half_integers} and of the Bessel function for $\nu=\frac{1}{2}$ \eqref{Bessel_1/2},  we can re-write formula \ref{formula:sym_c/n} as:
\begin{equation}\label{sym_c/n_p}
    C_{c/n} \, = \ 
    e^{-r\tau} \,
    \left[
    \, \frac{1}{2} \, + \, \frac{\alpha}{\pi}e^{\alpha\delta\tau} \, 
    \sum\limits_{p=0}^{\infty} \, \frac{(-1)^p k_0^{2p+1}}{p!(2p+1)} \, \K_{p+1}(\alpha\delta\tau) \,
    \left(  \frac{2\delta\tau}{\alpha}  \right)^{-p}
    \right]
    .
\end{equation}
The representation \eqref{sym_c/n_p} is less compact than formula \ref{formula:sym_c/n}, however it allows for a direct computation of the put option: indeed, using
\begin{equation}
    \mathbb{E}^{\mathbb{Q}} [ \mathbbm{1}_{ \{S_T>K\} } \, \vert \, S_t  ] 
    \, = \, 1 \, - \, \mathbb{E}^{\mathbb{Q}} [ \mathbbm{1}_{ \{S_T < K\} } \, \vert \, S_t  ]
    ,
\end{equation}
then it follows immediately from \eqref{sym_c/n_p} that the cash-or-nothing put can be written down as:
\begin{equation}\label{sym_c/n_put_p}
    P_{c/n} \, = \ 
    e^{-r\tau} \,
    \left[
    \, \frac{1}{2} \, - \, \frac{\alpha}{\pi}e^{\alpha\delta\tau} \, 
    \sum\limits_{p=0}^{\infty} \, \frac{(-1)^p k_0^{2p+1}}{p!(2p+1)} \, \K_{p+1}(\alpha\delta\tau) \,
    \left(  \frac{2\delta\tau}{\alpha}  \right)^{-p}
    \right]
    .
\end{equation}

\subsection{At the money approximations}\label{subsec:ATMF}
Let us assume throughout this subsection that options are at the money forward (ATMF), that is, $S_t=F$; retaining only the leading term of formula \ref{formula:sym_eur}, we can approximate the  European call by 
\begin{equation}\label{sym_eur_ATMF}
    C_{eur} \, \simeq \, 
    \frac{S_t\delta\tau e^{\alpha\delta\tau}}{\pi} \mathrm{K}_0 (\alpha\delta\tau)
    .
\end{equation}
Using the asymptotic behavior of the Bessel function for large arguments \eqref{Bessel_large_z}, we recover the fact that
\begin{equation}\label{BS_ATMF}
    C_{eur} \, \underset{\alpha\rightarrow\infty}{\longrightarrow} \,  \frac{S_t}{\sqrt{2\pi}} \, \sigma \, \sqrt{\tau}
    ,
\end{equation}
where $\sigma^2 : = \delta/\alpha$; \eqref{BS_ATMF} is the well-known approximation by \cite{Brenner94} for the ATMF Black-Scholes call. The approximation \eqref{sym_eur_ATMF} is also useful for parameter estimation: denoting by $C_t$ the market price of an ATMF European call option at time t and using Hankel's expansion \eqref{Bessel_large_z} up to  $k=1$, we obtain the quadratic equation 
\begin{equation}
    X^2 \, - \, \alpha\sqrt{2\pi} \frac{C_t}{S_t} X \, - \, \frac{1}{8}
    \, = \, 0
\end{equation}
where $X := \sqrt{\alpha\delta\tau}$. The positive solution reads
\begin{equation}
    X \, = \, \frac{1}{2} \left( \alpha\sqrt{2\pi}\frac{C_t}{S_t} 
    + \sqrt{2\pi\alpha^2\frac{C_t^2}{S_t^2} + \frac{1}{2}}  
    \right)
\end{equation}
and, therefore, using a Taylor expansion and turning back to the initial variables, we have
\begin{equation}\label{NIG_IV}
    \delta \, = \, \frac{2\pi\alpha}{\tau}\frac{C_t^2}{S_t^2} \, + \, 
    \frac{1}{4\alpha\tau} \, + \, O \left( \frac{1}{\alpha^3} \right)
    .
\end{equation}
Taking only the first order term in \eqref{NIG_IV}, we recover the ATMF value for the implied volatility $\sigma_I$ in the Black-Scholes model:
\begin{equation}
    \sigma_I \, := \, \sqrt{\frac{\delta}{\alpha}} \, = \, \sqrt{\frac{2\pi}{\tau}} \frac{C_t}{S_t}
    .
\end{equation}

\subsection{Miscellaneous payoffs}\label{subsec:miscellaneous}
In this subsection, we provide other applications of proposition \ref{prop:symm_factorization}, by considering path-independent payoffs featuring some more exotic attributes.

\paragraph{Gap option}
A gap (sometimes called pay-later) call has the following payoff:
\begin{equation}\label{Payoff_Gap}
    \mathcal{P}_{gap} (S_T,K_1,K_2) \, = \, (S_T-K_1) \mathbbm{1}_{ \{ S_T > K_2 \} }
\end{equation}
and degenerates into the European call when trigger and strike prices coincide ($K_1=K_2=K$). From the definition \eqref{Payoff_Gap}, it is immediate to see that the value at time $t$ of the Gap call is:
\begin{equation}
    C_{gap}  = C_{a/n} \, - \, K_1 \, C_{c/n}
\end{equation}
where the value of the asset-or-nothing and cash-or-nothing calls are given by formulas~\ref{formula:sym_A/N} and \ref{formula:sym_c/n} for $K=K_2$.

\paragraph{Power options}
Power options deliver a non linear payoff and are an easy way to increase the leverage ratio of trading strategies; the payoffs of the digital power calls are
\begin{equation}
    \mathcal{P}_{pow.c/n}(S_T,K) \, = \, \mathbbm{1}_{ \{S_T^a > K\} }
    \hspace{1cm}   
    \mathcal{P}_{pow.a/n}(S_T,K) \, = \, S_T^a \mathbbm{1}_{ \{S_T^a > K\} }
\end{equation}
for some $a>0$, and the power European call is:
\begin{equation}
    \mathcal{P}_{pow.eur}(S_T,K) 
    \, := \,
    \left[  S_T^a - K \right]^+
    .
\end{equation}
Introducing the notation
\begin{equation}
     k_a \, := \, \log\frac{S_t}{K^{\frac{1}{a}}} + (r - q + \omega) \tau 
     \, , \hspace{0.5cm} k_{0,a} \, := \, k_a  +  \mu\tau
\end{equation}
then we can remark that:
\begin{equation}
    \mathcal{P}_{pow.a/n}(S_t e^{(r-q+\omega)\tau + x},K) \, = \, Ke^{a(k_a+x)} \, \mathbbm{1}_{ \{ x > -k_a \} } 
    .
\end{equation}
Therefore, using the representations (see table \ref{tab:Mellin} in appendix \ref{app:Mellin})
\begin{equation}
    e^{ a (k_a +x)} \, = \, \int\limits_{c_2 - i\infty}^{c_2 + i\infty} (-1)^{-s_
    2} a^{-s_2} \Gamma(s_2) (k_a  + x)^{-s_2} \, \frac{\ud s_2}{2i\pi} \hspace{1cm} (c_2 > 0)
\end{equation}
and
\begin{equation}
    e^{ a (k_a +x)}  - 1 \, = \, \int\limits_{c_2 - i\infty}^{c_2 + i\infty} (-1)^{-s_
    2} a^{-s_2} \Gamma(s_2) (k_a  + x)^{-s_2} \, \frac{\ud s_2}{2i\pi} \hspace{1cm} (-1 < c_2 < 0)
\end{equation}
and proceeding exactly the same way than for proving formulas \ref{formula:sym_A/N}, \ref{formula:sym_eur} and \ref{formula:sym_c/n}, we obtain:
\begin{formula}[Power options]\label{form:pow}
    The values at time $t$ of the power options are:
    \begin{itemize}
        \item[-] Asset-or-nothing power call:
        \begin{equation}\label{sym_pow_a/n}
            C_{pow.a/n} \, = \, \frac{ K\alpha e^{(\alpha \delta - r) \tau}}{\sqrt{\pi}} \, 
            \sum\limits_{\substack{n_1 = 0 \\ n_2 = 0}}^{\infty} \,
            \frac{ a^{n_2} k_{0,a}^{n_1} }{n_1! \Gamma( 1+\frac{-n_1+n_2}{2} )} \,
            \K_{\frac{n_1-n_2+1}{2}}(\alpha\delta\tau)
            \left(  \frac{\delta\tau}{2\alpha} \right)^{\frac{-n_1+n_2+1}{2}}
            ;
        \end{equation} 
        \item[-] European power call:
        \begin{equation}\label{sym_pow_eur}
            C_{pow.eur} \, = \, \frac{ K\alpha e^{(\alpha \delta - r) \tau}}{\sqrt{\pi}} \, 
            \sum\limits_{\substack{n_1 = 0 \\ n_2 = 1}}^{\infty} \,
            \frac{ a^{n_2} k_{0,a}^{n_1} }{n_1! \Gamma( 1+\frac{-n_1+n_2}{2} )} \,
            \K_{\frac{n_1-n_2+1}{2}}(\alpha\delta\tau)
            \left(  \frac{\delta\tau}{2\alpha} \right)^{\frac{-n_1+n_2+1}{2}}
            ;
            \end{equation}  
        \item[-] Cash-or-nothing power call:
        \begin{equation}\label{sym_pow_c/n}
            C_{pow.c/n} \, = \, \frac{  \alpha e^{(\alpha \delta - r) \tau}}{\sqrt{\pi}} \, 
            \sum\limits_{n = 0}^{\infty} \,
            \frac{ k_{0,a}^{n} }{n! \Gamma( 1-\frac{n}{2} )} \,
            \K_{\frac{n+1}{2}}(\alpha\delta\tau)
            \left(  \frac{\delta\tau}{2\alpha} \right)^{\frac{-n+1}{2}}
            .
        \end{equation}
    \end{itemize}
\end{formula}
It is clear that, for the series \eqref{sym_pow_a/n}, \eqref{sym_pow_eur} and \eqref{sym_pow_c/n} to converge, assumption \ref{ass:growth} has to be satisfied by $k_{0,a}$ and no longer by $k_0$, that is:
\begin{equation}
    \left\vert  \frac{k_{0,a}}{\delta\tau}  \right\vert \, < \, 1
    .
\end{equation}

\paragraph{Log options, log contract}
Log options are, basically, options on the rate of return of the underlying (\cite{Wilmott06}). The payoff of a log call and of a log put are:
\begin{equation}
    \mathcal{P}_{log \, call}(S_T, K) \, := \, [\log S_T \, - \log K]^+ \, ,
    \hspace{1cm}
    \mathcal{P}_{log \, put}(S_T, K) \, := \, [\log K \, - \log S_T]^+ 
    .
\end{equation}
The log contract, introduced by \cite{Neuberger94}, is a forward contract that is obtained by being long of a log call and short of a log put, resulting in
\begin{equation}
    \mathcal{P}_{log \, contract}(S_T, K) \, = \, \log \frac{S_T}{K}
    .
\end{equation}
Note that a delta-hedged log contract with $K=1$ is actually a synthetic variance swap: indeed, by denoting the quadratic variation of $S$ by $< S >$ and using It\^{o}'s lemma, it is well known that, in the Black-Scholes model,
\begin{equation}\label{payoff_variance_swap}
    \mathbb{E}^{\mathbb{Q}} \left[ <S>_T \, - \, <S>_t  \, \vert \, S_t \right]
    \, = \, 
    2 \, 
    \mathbb{E}^{\mathbb{Q}} \left[ -\log\frac{S_T}{S_t} \, + \, \frac{S_T}{S_t} \ - 1  \, \vert \, S_t \right]
    .
\end{equation}
In the more general framework of exponential L\'evy models, the overall multipliers in the r.h.s. of \eqref{payoff_variance_swap} are different from 2 and have been determined in \cite{Carr12}; for instance in the symmetric NIG models, it is equal to $\frac{1}{\alpha(\alpha-\sqrt{\alpha^2-1})}$ which, as expected, tends to $2$ when $\alpha\rightarrow\infty$ . Let us therefore show how to derive pricing formulas for the log options and the log contract in this model: remarking that, using notations \eqref{moneyness},
\begin{equation}
    \mathcal{P}_{log\,call} (S_t e^{(r-q+\omega)\tau + x} , K) \, = \, [k + x]^+ 
    ,
\end{equation}
it follows that the Mellin transform for the payoff function \eqref{sym_payoff_transform} reads, for the log call:
\begin{equation}
    P^{*}(s_1) \, = \, \int\limits_{-k}^{\infty} \, (k+x) \, (x-\mu\tau)^{-s} \, \ud x \, = \, \frac{(-k_0)^{2-s_1}}{(s_1-2)(s_1-1)} 
\end{equation}
and, using proposition \ref{prop:symm_factorization}, that the log call price itself writes:
\begin{equation}
    C_{log} \, = \, \frac{\alpha e^{(\alpha\delta-r)\tau}}{2\pi}
    \,
    \int\limits_{c_1-i\infty}^{c_1+i\infty} \, 
    \frac{\Gamma(\frac{s_1}{2})}{(s_1-2)(s_1-1)} \,
    (-k_0)^{2-s_1} \, 
    \K_{1-\frac{s_1}{2}}(\alpha\delta\tau) \,
    \left( \frac{2\delta\tau}{\alpha} \right)^{\frac{s_1}{2}} \,
    \frac{\ud s_1}{2i\pi}
\end{equation}
where $c_1 > 2$. Similarly, the log put writes:
\begin{equation}
    P_{log} \, = \, \frac{\alpha e^{(\alpha\delta-r)\tau}}{2\pi}
    \,
    \int\limits_{c_1-i\infty}^{c_1+i\infty} \, 
    \frac{\Gamma(\frac{s_1}{2})}{(s_1-2)(s_1-1)} \,
    k_0^{2-s_1} \, 
    \K_{1-\frac{s_1}{2}}(\alpha\delta\tau) \,
    \left( \frac{2\delta\tau}{\alpha} \right)^{\frac{s_1}{2}} \,
    \frac{\ud s_1}{2i\pi}
\end{equation}
where $c_1 > 2$. Summing all residues arising at $s_1 = 2$, $s_1=1$ and $s_1=-2n$, $n\in\mathbb{N}$, grouping the terms and simplifying yields:
\begin{formula}[Log options, log contract]\label{form:log}
    The value at time $t$ of a log option is:
    \begin{itemize}
        \item[-] Log call:
        \begin{equation}\label{sym_log_call}
            C_{log} \, = \, e^{-r\tau} \,
            \left[
            \frac{k_0}{2} \, + \,
            \frac{\alpha e^{\alpha\delta\tau}}{2\pi} \,
            \sum\limits_{n=0}^{\infty} \, 
            \frac{(-1)^{n-1} k_0^{2n}}{n!(2n-1)} \, 
            \K_n (\alpha\delta\tau) \,
            \left( \frac{2\delta\tau}{\alpha}  \right)^{-n+1}
            \right]
            ;
            \end{equation}
        \item[-] Log put:
        \begin{equation}\label{sym_log_put}
            P_{log} \, = \, e^{-r\tau} \,
            \left[
            -\frac{k_0}{2} \, + \,
            \frac{\alpha e^{\alpha\delta\tau}}{2\pi} \,
            \sum\limits_{n=0}^{\infty} \, 
            \frac{(-1)^{n-1} k_0^{2n}}{n!(2n-1)} \, 
            \K_n (\alpha\delta\tau) \,
            \left( \frac{2\delta\tau}{\alpha}  \right)^{-n+1}
            \right] 
            ;
            \end{equation}  
        \item[-] Log contract:
        \begin{equation}\label{sym_log_contract}
            C_{log} \, - \, P_{log} \, = \, 
            e^{-r\tau} \, k_0
            .
        \end{equation}
    \end{itemize}
\end{formula}
Recall that, when $\alpha\rightarrow\infty$, $\omega\sim -\mu-\frac{\sigma^2}{2}$, where $\sigma^2 : = \frac{\delta}{\alpha}$  and therefore the log contract \eqref{sym_log_contract} becomes 
\begin{equation}
    \left( C_{log} \, - \, P_{log}  \right)^{(\alpha\rightarrow\infty)} \, = \, e^{-r\tau} \, 
    \left( \log\frac{S_t}{K} \, + \, (r-q-\frac{\sigma^2}{2} ) \tau
    \right)
\end{equation}
which, taking $K=1$, is the formula originally obtained by \cite{Neuberger94} for the price of a log contract in the Black-Scholes model.

\paragraph{Capped payoffs} Suppose that we wish introduce a cap to limit the exercise range of a digital option for example; in this case, the payoff of the cash-or-nothing call would read:
\begin{equation}\label{payoff_c/n_capped}
    \mathcal{P}_{capped \, c/n} (S_T,K_-,K_+) \, := \, \mathbbm{1}_{ \{ K_- < S_T < K_+  \} }
\end{equation}
where $K_-$ is the strike price, and $K_+$ the cap. It is clear that \eqref{payoff_c/n_capped} can be decomposed into the difference of two cash-or nothing calls with strike prices $K_-$ and $K_+$. Therefore, introducing the notations 
\begin{equation}
     k_\pm \, := \, \log\frac{S_t}{K_\pm} + (r - q + \omega) \tau 
     \, , \hspace{0.5cm} k_{0,\pm} \, := \, k_\pm  +  \mu\tau
\end{equation}
then it follows immediately from formula \ref{formula:sym_c/n} that the value at time t of the capped cash-or-nothing call is given by
 \begin{equation}\label{sym_cap_c/n}
        C_{capped \, c/n} \, = \, \frac{  \alpha e^{(\alpha \delta - r) \tau}}{\sqrt{\pi}} \, 
        \sum\limits_{n = 0}^{\infty} \,
        \frac{ k_{0,-}^n - k_{0,+}^n  }{n! \Gamma( 1-\frac{n}{2} )} \,
        \K_{\frac{n+1}{2}}(\alpha\delta\tau)
        \left(  \frac{\delta\tau}{2\alpha} \right)^{\frac{-n+1}{2}}
        .
\end{equation} 
Of course, for \eqref{sym_cap_c/n} to converge, one needs assumption \eqref{ass:growth} to be satisfied for both $k_{0,-}$ and $k_{0,+}$. Extension to the case of an option activated outside the interval $[K_-,K_+]$ is straightforward, by writing down:
\begin{equation}
    \mathbb{E}^{\mathbb{Q}} \left[\mathbbm{1}_{ \{S_T < K_-\} \cup \{S_T > K_+\}  } \, \vert \,  S_t \right]
    \,= \,
    1 \, - \, \mathbb{E}^{\mathbb{Q}} \left[ \mathbbm{1}_{ \{ K_- < S_T < K_+  \} } \, \vert \,  S_t \right]
\end{equation}
and by using \eqref{sym_cap_c/n}.

\section{Option pricing in the asymmetric model}\label{sec:asymmetric}

Let us now consider the case where the process $\{ X_t \}_{t\in[0,T]}$ in \eqref{Exp_NIG_SDE} is distributed according to the asymmetric distribution $\mathrm{NIG}(\alpha, \beta , \delta t, \mu t)$, $\beta\neq 0$. All notations defined in \eqref{moneyness} remain valid, but we introduce the supplementary definition $\gamma \, := \, \sqrt{\alpha^2 - \beta^2}$, such that $k_0$ can be written down as:
\begin{equation}
    k_0 \, = \, \log\frac{S_t}{K} \, + \, 
    \left(r - q + \delta \left( \sqrt{\alpha^2 - (\beta+1)^2} - \gamma \right) \right) \tau 
    .
\end{equation}
To simplify the notations, as multiple $\mathbb{C}$-integrals will be involved, we will denote the vectors in $\mathbb{C}^n$ by $\underline{z}:=^t [z_1 , \dots , z_n]$, $z_i\in\mathbb{C}$ for $i= 1\dots n$, and we will use the notation 
\begin{equation}
    \underline{c} \, + \, i\mathbb{R}^n \, := \, (c_1 + i\mathbb{R} ) \, \times \, (c_2 + i\mathbb{R} ) \, \, \dots \, \times \, (c_n + i\mathbb{R} )
    .
\end{equation}

\subsection{Pricing formula}\label{subsec:asym_pricing}

Like in section \ref{sec:symmetric}, we start by establishing a representation for the NIG density $f(x,t)$ under the form of a Mellin-Barnes integal, but this time in the asymmetric case.

\begin{lemma}\label{lemma:asym_density}
    For any $\underline{c} \in \mathbb{R}_+^2$, the following holds true:
    \begin{multline}\label{asym_density}
        f(x,t) \, = \, \frac{\alpha}{2\pi} \, e^{\gamma \delta t} \, 
        \\
        \times \int\limits_{\underline{c} + i\mathbb{R}^2} \, 
        (-1)^{-s_2}\beta^{-s_2} \, 
        \Gamma\left( \frac{s_1}{2} \right) \,
        \Gamma(s_2) \, 
        \K_{1-\frac{s_1}{2}}(\alpha\delta t) \,
        \left( \frac{2\delta t}{\alpha} \right)^{\frac{s_1}{2}} \,
        \vert x - \mu t\vert ^{-s_1} \,
        (x-\mu t)^{-s_2}
        \frac{\ud s_1 \ud s_2}{(2i\pi)^2}
        .
    \end{multline}
\end{lemma}
\begin{proof}
    Like in the proof of lemma \ref{lemma:sym_density}, we introduce the Mellin representation \eqref{BesselK_Mellin} for the Bessel function that holds for $c_1\in\mathbb{R}$, and we introduce a supplementary representation for the exponential term (see table \ref{tab:Mellin} in appendix \ref{app:Mellin}): 
    \begin{equation}\label{asym_exp}
        e^{\beta(x-\mu t)} \, = \, 
        \int\limits_{c_2-i\infty}^{c_2+i\infty} \, (-1)^{-s_2} \, \beta^{s_2} \, \Gamma(s_2) \, (x - \mu t)^{-s_2} \, \frac{\ud s_2}{2i\pi}
    \end{equation}
    that holds for $c_2\in\mathbb{R}_+$. Inserting \eqref{BesselK_Mellin} and \eqref{asym_exp} into the density \eqref{NIG_process_density} yields the reprensentation \eqref{asym_density}.
\end{proof}

Let us now introduce the asymmetric analogue to the $P^*(s_1)$ function \eqref{sym_payoff_transform}:
\begin{equation}\label{asym_payoff_transform}
    P^*(s_1,s_2) \, = \, \int\limits_{-\infty}^{\infty} \, \mathcal{P} \left( S_t e^{(r - q+\omega)\tau +x} , K_1, \dots , K_n \right) \, |x - \mu\tau|^{-s_1} \, (x-\mu\tau)^{-s_2} \, \ud x
\end{equation}
and assume that it exists for $( Re(s_1), Re(s_2)) \in P$ for a certain subset $P\subset\mathbb{R}^2$. Then, as a consequence of the risk-neutral pricing formula \eqref{Risk-neutral_2} and of lemma \ref{lemma:asym_density}, we immediately obtain:

\begin{proposition}[Factorization in the Mellin space]
    \label{prop:asymm_factorization}
    Let $\underline{c}\in \tilde{P}$ where $\tilde{P} := P \cap \mathbb{R}_+^2$ is assumed to be nonempty. Then the value at time $t$ of a contingent claim delivering a payoff $\mathcal{P}(S_T,K_1, \dots , K_n)$ at its maturity $t=T$ is equal to:
    \begin{equation}\label{asym_factorization}
        \mathcal{C}  \, = \, 
        \frac{\alpha}{2\pi} \, e^{(\gamma \delta - r) \tau} \,
        \int\limits_{\underline{c} + i\mathbb{R}^2} \, 
        (-1)^{-s_2}\beta^{-s_2}
        \Gamma\left( \frac{s_1}{2} \right) \,
        \Gamma(s_2) \,
        P^*(s_1,s_2) \,
        \K_{1-\frac{s_1}{2}}(\alpha\delta \tau) \,
        \left( \frac{2\delta \tau}{\alpha} \right)^{\frac{s_1}{2}}  \, \frac{\ud s_1 \ud s_2}{(2i\pi)^2}
        .
    \end{equation}
\end{proposition}

\subsection{Digital and European options}\label{subsec:asym_dig_eur}

To illustrate some applications of proposition \ref{prop:asymm_factorization}, we compute the price of the digital and European options, whose payoffs were defined in subsection \ref{subsec:sym_dig_eur}. We also recall the notation for the Pochhammer symbol $(a)_n:=\frac{\Gamma(a+n)}{\Gamma(a)}$.

\begin{formula}[Asset-or-nothing call] 
    \label{formula:asym_A/N}
    The value at time $t$ of an asset-or-nothing call option is: \\
    \begin{equation}\label{asym_A/N}
        C_{a/n} \, = \, \frac{ K\alpha e^{(\gamma \delta - r) \tau}}{\sqrt{\pi}} \, 
        \sum\limits_{n_1, n_2, n_3 = 0}^{\infty} \,
        \frac{(-n_1+n_3+1)_{n_2} \, k_0^{n_1} \beta^{n_2}}{n_1! n_2! \Gamma(1 + \frac{-n_1+n_2+n_3}{2})} \,
        \K_{\frac{n_1-n_2-n_3+1}{2}}(\alpha\delta\tau) \,
        \left(  \frac{\delta\tau}{2\alpha} \right)^{\frac{-n_1+n_2+n_3+1}{2}}
        .
    \end{equation} 
\end{formula}
\begin{proof}
    \underline{Step 1:} The proof starts like the proof of formula \ref{formula:sym_A/N}, by assuming $k_0 < 0$, by remarking that
    \begin{equation}
        \mathcal{P}_{a/n} \left( S_t e^{(r-q+\omega)\tau +x} , K \right)
        \, = \,
        K \, e^{k+x} \, \mathbbm{1}_{ \{ x > - k \} }
        .
    \end{equation}
    and by introducing a Mellin-Barnes representation for the exponential term in the option's payoff:
        \begin{equation}\label{MB_payoff_a/n_s3}
        e^{k +x} \, = \, \int\limits_{c_3 - i\infty}^{c_3 + i\infty} (-1)^{-s_
        3} \Gamma(s_3) (k  + x)^{-s_3} \, \frac{\ud s_3}{2i\pi} \hspace{1cm} (c_3 > 0)
        .
        \end{equation}
    Therefore, the $P^*(s_1,s_2)$ function \eqref{asym_payoff_transform} reads:
    \begin{align}\label{asym_A/N_payoff_Mellin}
        P^*(s_1,s_2) & = \, K \int\limits_{c_3 - i\infty}^{c_3 + i\infty} (-1)^{-s_3} \Gamma(s_3) 
        \int\limits_{ -k }^\infty (k + x)^{-s_3} (x - \mu\tau)^{-s_1-s_2} \, \ud x \, \frac{\ud s_3}{2i\pi} \\
        & = K \int\limits_{c_3 - i\infty}^{c_3 + i\infty} (-1)^{-s_2} \frac{\Gamma(s_3)\Gamma(1-s_3)\Gamma(s_1+s_2+s_3-1)}{\Gamma(s_1+s_2)} (-k_0)^{-s_1-s_2-s_3+1}   
        \, \frac{\ud s_3}{2i\pi}
    \end{align}
    where the $x$-integral exists because $k_0<0$. Using proposition \ref{prop:asymm_factorization}, we obtain the price of the asset-or-nothing call:
    \begin{multline}
        \mathcal{C}  \, = \, 
        \frac{K \alpha}{2\pi} \, e^{(\gamma \delta - r) \tau} 
        \int\limits_{\underline{c} + i\mathbb{R}^3} \, 
        (-1)^{-s_2-s_3}\beta^{-s_2} \,
        \frac{\Gamma(\frac{s_1}{2})\Gamma(s_2)\Gamma(s_3)\Gamma(1-s_3)\Gamma(s_1+s_2+s_3-1)}{\Gamma(s_1+s_2)} \, 
        (-k_0)^{-s_1-s_2-s_3+1}
        \\
        \times \K_{1-\frac{s_1}{2}}(\alpha\delta\tau) \,
        \left( \frac{2\delta \tau}{\alpha} \right)^{\frac{s_1}{2}}  \, \frac{\ud s_1 \ud s_2 \ud s_3}{(2i\pi)^3}
    \end{multline}
    which converges in the subset $\{ (s_1,s_2,s_3)\in\mathbb{C}^3, Re(s_1)>0, Re(s_2)>0, 0 < Res(s_3) < 1, Re(s_1+s_2+s_3)>1  \}$ and can be analytically continued outside this polyhedron, except when the Gamma functions in the numerator are singular.  If we consider the singularities induced by $\Gamma(s_2)$ at $s_2 = -n_2$, $n_2\in\mathbb{N}$, by $\Gamma(s_3)$ at $s_3 = -n_3$, $n_3\in\mathbb{N}$ and by $\Gamma(s_1+s_2+s_3-1)$ at $s_1+s_2+s_2-1 = -n_1$, $n_1\in\mathbb{N}$, then, the associated residues are straightforward to compute via the change of variables $u:=s_1+s_2+s-3-1$, $v:=s_2$, $w=s_3$ and via the singular behavior \eqref{sing_Gamma} for the Gamma functions; they read:
    \begin{multline}\label{asym_A/N_res}
        \frac{ K\alpha e^{(\gamma \delta - r) \tau}}{2\pi}
        (-1)^{n_2+n_3}
        \beta^{n_2}
        \frac{(-1)^{n_1}}{n_1!} \frac{(-1)^{n_2}}{n_2!} \frac{(-1)^{n_3}}{n_3!}
        \frac{\Gamma(1+n_3)\Gamma(\frac{-n_1+n-2+n_3+1}{2})}{\Gamma(-n_1+n_3+1)}
        (-k_0)^{n_1}
        \\
        \times \K_{1-\frac{-n_1+n_2+n_3+1}{2}}(\alpha\delta\tau)
        \left( \frac{\delta\tau}{2\alpha} \right)^{\frac{-n_1+n_2+n_3+1}{2}}
        .
    \end{multline}
    Using the Legendre duplication formula \eqref{Legendre} and the definition of the Pochhammer symbol \eqref{Pochhammer_def}, we write:
    \begin{equation}
        \frac{\Gamma(\frac{-n_1+n_2+n_3+1}{2})}{\Gamma(-n_1+n_3+1)}
        \, = \, 
        \frac{\sqrt{\pi}}{2^{-n_1+n_2+n_3}}  \, \frac{(-n_1+n_3+1)_{n_2}}{\Gamma(1+\frac{-n_1+n_2+n_3}{2})}
        .
    \end{equation}
    Inserting into \eqref{asym_A/N_res}, simplifying and summing all residues for $n_1,n_2,n_3\in\mathbb{N}$ yields the series \eqref{asym_A/N}.
    
    \noindent \underline{Step 2:} Like in the proof of formula \ref{formula:sym_A/N}, extension to the case $k_0>0$ is performed thanks to the parity
        \begin{equation}
        \mathbb{E}^\mathbb{Q} [ S_T \, \mathbbm{1}_{ \{ S_T > K \} }  \, \vert \, S_t ]
        \, = \, 
        S_t \, e^{(r-q)\tau} \, - \, 
        \mathbb{E}^\mathbb{Q} [ S_T \, \mathbbm{1}_{ \{ S_T < K \} }  \, \vert \, S_t ]
        .
    \end{equation}
    
    \noindent\underline{Step 3:} Last, using the same estimates than in the proof of formula \ref{formula:sym_A/N}, the series \eqref{asym_A/N_res} converges when $n_2,n_3\rightarrow\infty$ for all parameter values, and when $n_1\rightarrow\infty$ if and only if assumption \ref{ass:growth} is satisfied.
\end{proof}

\begin{formula}[European call] 
    \label{formula:asym_eur}
    The value at time $t$ of a European call option is: \\
    \begin{equation}\label{asym_eur}
        C_{eur} \, = \, \frac{ K\alpha e^{(\gamma \delta - r) \tau}}{\sqrt{\pi}} \, 
        \sum\limits_{\substack{n_1, n_2 =0 \\ n_3 = 1}}^{\infty} \,
        \frac{(-n_1+n_3+1)_{n_2} \, k_0^{n_1} \beta^{n_2}}{n_1! n_2! \Gamma(1 + \frac{-n_1+n_2+n_3}{2})} \,
        \K_{\frac{n_1-n_2-n_3+1}{2}}(\alpha\delta\tau) \,
        \left(  \frac{\delta\tau}{2\alpha} \right)^{\frac{-n_1+n_2+n_3+1}{2}}
        .
    \end{equation} 
\end{formula}
\begin{proof}
Like in the proof of formula \ref{formula:sym_eur}, we remark that we can write:
    \begin{equation}
        \mathcal{P}_{eur} (Se^{(r - q +\omega)\tau+x},K) 
        \, = \,
        K (e^{k + x} -1) \mathbbm{1}_{ \{x > -k \} }
        .
    \end{equation}
    Then, we use the Mellin-Barnes representation (see table \ref{tab:Mellin} in appendix \ref{app:Mellin}):
    \begin{equation}\label{MB_payoff_european_s3}
        e^{k + x} - 1 \, = \, \int\limits_{c_3 - i\infty}^{c_3 + i\infty} (-1)^{-s_
        3} \Gamma(s_3) (k  + x)^{-s_3} \, \frac{\ud s_3}{2i\pi} \hspace{1cm} (-1 < c_3 < 0)
    \end{equation}
    and we proceed exactly the same way than for proving Formula \ref{formula:asym_A/N}; the $n_3$-summation in \eqref{asym_eur} starts in $n_3=1$ instead of $n_3=0$, because the strip of convergence of \eqref{MB_payoff_european_s3} is reduced to $<-1,0>$ instead of $<0,\infty>$ in \eqref{MB_payoff_a/n_s3}.
\end{proof}
By difference of \eqref{asym_A/N} and \eqref{asym_eur}, we immediately obtain the formula for the cash-or-nothing call:
\begin{formula}[Cash-or-nothing call] 
    \label{formula:asym_c/n}
    The value at time $t$ of a cash-or-nothing call option is: \\
    \begin{equation}\label{asym_c/n}
        C_{c/n} \, = \, \frac{ \alpha e^{(\gamma \delta - r) \tau}}{\sqrt{\pi}} \, 
        \sum\limits_{n_1, n_2 =0}^{\infty} \,
        \frac{(-n_1+1)_{n_2} \, k_0^{n_1} \beta^{n_2}}{n_1! n_2! \Gamma(1 + \frac{-n_1+n_2}{2})} \,
        \K_{\frac{n_1-n_2+1}{2}}(\alpha\delta\tau) \,
        \left(  \frac{\delta\tau}{2\alpha} \right)^{\frac{-n_1+n_2+1}{2}}
        .
    \end{equation} 
\end{formula}

\section{Applications and numerical tests}\label{sec:num}

In this section, we show how to implement very simply the pricing formulas we derived, for instance in an Excel spreadsheet. we also determine what restriction is induced by assumption \ref{ass:growth} in terms of accessible option maturities, and we provide some precise estimates for the convergence speed and the truncation errors of the series. Then, we compare the various pricing formulas established in the above with several numerical tools, and demonstrate the reliability and efficiency of the results.

\subsection{Practical implementation}\label{subsec:practical}

Let us show on some examples how our pricing formulas can be used in practice. In figure \ref{fig:Excel}, we have implemented the formulas \ref{formula:sym_A/N}, \ref{formula:sym_eur} and \ref{formula:sym_c/n} in an Excel spreadsheet, up to $n_1=n_2=10$. This can be done in a straightforward way, thanks to the functions BESSELK(x,n) and GAMMA(x). The red square yields the price of the asset-or-nothing call $C_{a/n}$, and the blue rectangle starting at $n_2=1$ yields the price of the European call $C_{eur}$; the remaining green rectangle for $n_2=0$ represents the price of the cash-or-nothing call multiplied by the strike price, i.e. $K\times C_{c/n}$. 

\begin{figure}[H]
\centering
\includegraphics[scale=0.6]{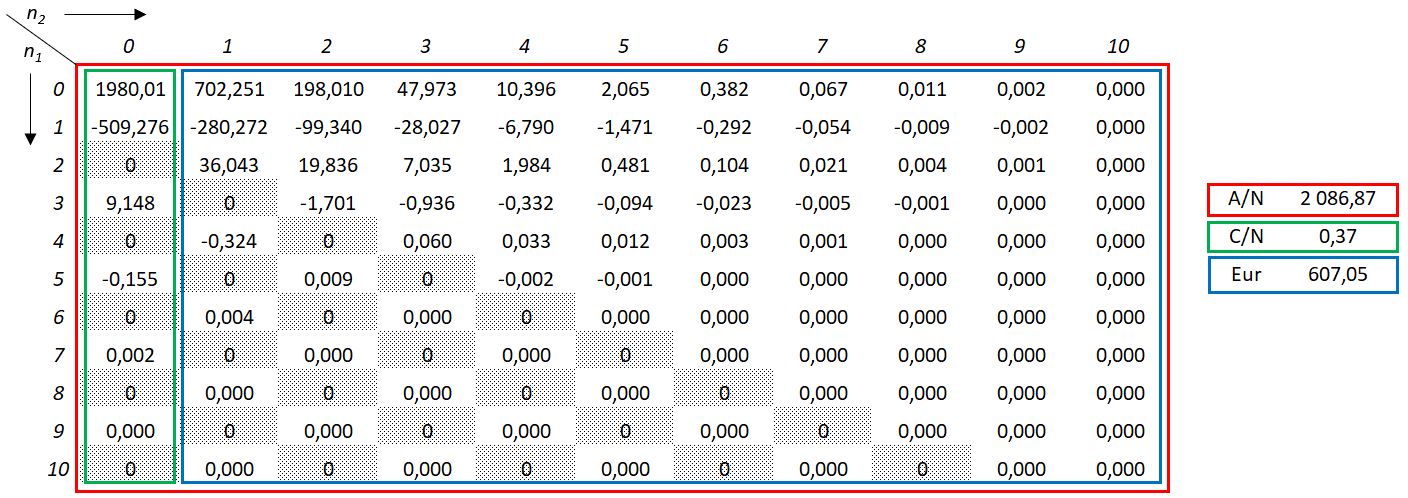}
\caption{Series terms for the formulas \ref{formula:sym_A/N} (red), \ref{formula:sym_eur} (blue) and \ref{formula:sym_c/n} (green). Choice of parameters: $S_t$=3800, $K=4000$, $r = 1\%$, $q = 0\%$, $\tau = 1$, $\alpha=10$, $\delta = 2$.}
\label{fig:Excel}
\end{figure}

We may first note that the convergence is extremely fast; moreover, computations are even more accelerated by two properties:
\begin{itemize}
    \item[-] The symmetry property of the Bessel function  \eqref{Bessel_sym}, $ \K_{\nu} (z) \, = \, \K_{-\nu} (z)$, which reduces by two the number of Bessel functions to evaluate (but these evaluations are straightforward anyway);
    \item[-] The presence of several null series terms (in grey in fig. \ref{fig:Excel}), due to the divergence of the Gamma function in the denominator when its argument is a negative integer. It is not complicated to see that, if e.g. $n_1$ is an even integer $n_1=2p$, then there are $2\times\sum\limits_{k=1}^{p-1}=p(p-1)$ null terms. In fig. \ref{fig:Excel}, we choose $n_1=10$, corresponding to $11\times 10=110$ terms in the computation of the European call price, and an attained precision of $10^{-2}$; but as there are $5\times 4=20$ null terms in the series, only $90$ terms are actually needed to attain this precision.
 \end{itemize}

\subsection{Accessible range of parameters}


We start by remarking that the at the money (ATM) situation ($S_t=K$) is a favorable situation for satisfying assumption \ref{ass:growth}. Indeed, in that case, we have:
\begin{equation}\label{rho_ATMF}
    \frac{\vert k_0 \vert}{\delta\tau} \, = \, \left\vert  \frac{r-q}{\delta}  \, + \, \sqrt{\alpha^2-(\beta+1)^2} - \sqrt{\alpha^2 \, - \, \beta^2} \right\vert
    .
\end{equation}
In the symmetric model in particular, it is clear that
\begin{equation}
    -1 + \frac{r-q}{\delta} \, < \,    \frac{r-q}{\delta}  \, + \, \sqrt{\alpha^2-1^2} - \alpha  \, < \, \frac{r-q}{\delta} 
\end{equation}
and therefore assumption \ref{ass:growth} is satisfied as soon as $r-q < \delta$; according to the implied parameters in table \ref{tab:accessible}, the smallest calibrated value for $\delta$ is $0.2483$, therefore assumption \ref{ass:growth} is satisfied (independently of $\alpha$ and of other market parameters) as soon as the risk-free interest rate is smaller than $25\%$, which is of course the case for most financial applications.

In the more general non at the money and non symmetric case, satisfying assumption \ref{ass:growth} necessitates some restriction on the option's maturities, depending on the moneyness situation. Assuming that $\mu=0$ (as option prices are not sensitive to $\mu$) and, introducing
\begin{equation}
    \rho_{\pm} \, := \, \frac{\log\frac{S_t}{K}}{\pm \delta - r + q - \omega}
    \, ,
\end{equation}
then it is not hard to see that:
\begin{itemize}
    \item[-] If $S_t>K$ (in the money (ITM) situation), then assumption \ref{ass:growth} is satisfied if $\tau > \rho_+$ or $\tau < \rho_-$;

    \item[-] If $S_t<K$ (out of the money (OTM) situation), then assumption \ref{ass:growth} is satisfied if $\tau > \rho_-$ or $\tau < \rho_+$.
\end{itemize}
In table \ref{tab:accessible}, we illustrate this rule on several implied NIG parameters, calibrated in the literature on various option markets: OBX options in \cite{Saebo09}, S\&P 500 options in \cite{Matsuda06,Albrecher05} or Euro Stoxx 50 (SX5E) options in \cite{Schoutens04}.

\begin{table}[ht]
 \caption{Maturities allowing that assumption \ref{ass:growth} is satisfied, for some sets of implied NIG parameters. Other parameters: $K=4000$, $r =1\%$, $q=0\%$ and $S_t=3500$ (OTM) or $S_t=4500$ (ITM).}
 \label{tab:accessible}       
 \centering
 \begin{tabular}{lccccc}
 \hline
 & \multicolumn{3}{c}{ NIG parameters  } &  \multicolumn{2}{c}{ Accessible maturities  }  \\
 & $\alpha$ & $\beta$ & $\delta$ & OTM & ITM \\
 \hline
 \cite{Saebo09} & 8.9932 & -4.5176 & 1.1528 & $\tau > 0.077$ & $\tau > 0.208$ \\
 \cite{Matsuda06} & 20.7408 & -11.7308 & 0.2483 & $\tau > 0.319$ & $\tau > 1.504$
 \\
 \cite{Schoutens04} & 16.1975 & -3.1804 & 1.0867 & $\tau > 0.104$ & $\tau > 0.131$
   \\
 \cite{Albrecher05} & 18.4815 & -4.8412 & 0.4685 & $\tau > 0.226$ & $\tau > 0.341$
 \\
 \hline
\end{tabular}
\end{table}

\subsection{Truncation error}

In this subsection we estimate the rest of some series arising in our pricing formulas, in order to determine what truncation has to be applied to obtain a desired level of precision in option prices. For simplicity of notations, we perform the analysis in the symmetric model, but extension to the asymmetric case is straightforward.

\paragraph{Cash-or-nothing} Let us observe that the general term of the cash-or-nothing series \eqref{sym_c/n_p} is the same than the $R_{2p+1}$ term introduced in \eqref{R_p} in the proof of formula \ref{formula:sym_A/N}:
\begin{equation}
    R_{2p+1} \, := \,  \frac{1}{\sqrt{\pi}}\frac{1}{2p+1}\frac{(-1)^p}{4^p p!} \, k_0^{2p+1}K_{p+1}(\alpha\delta\tau)\left(  \frac{\delta\tau}{2\alpha} \right)^{-p}  
        .
\end{equation}
Using the bound \eqref{R_p_infty}, we therefore know that, for $\epsilon>0$, there exists a rank $p_\epsilon$ such that the general term of the series in the cash-or-nothing formula \eqref{sym_c/n_p} is bounded by 
\begin{equation}\label{R_p_epsilon}
     |R_{2 p_\epsilon +1}| \, \sim \,  
     \left\vert \frac{1}{\sqrt{2\pi p_\epsilon (2p_\epsilon +2)}} \, \frac{k_0}{e\alpha\delta\tau} \, \left( \frac{k_0^2}{(\delta\tau)^2} \right)^{p_\epsilon} \right\vert
     \, < \, \epsilon
     .
\end{equation}
As a consequence of assumption \ref{ass:growth}, $| \frac{k_0}{e\delta\tau} | < 1$ and therefore, denoting by $\ceil{X}$ the least integer greater or equal to a real number $X$, it suffices to choose
\begin{equation}\label{p_epsilon}
    p_\epsilon \, = \, \ceil*{ \frac{\log\alpha\epsilon}{2 \log \vert \frac{k_0}{\delta\tau} \vert }  }
\end{equation}
to be sure that all terms of order $p \geq p_\epsilon$ are $O(\epsilon)$ in the series \eqref{sym_c/n_p}. Turning back to the $n$-variable (i.e. $n=2p+1$), it follows from \eqref{p_epsilon} that, definying
\begin{equation}\label{n_epsilon}
    n_\epsilon \, := \,   2 p_\epsilon \, + \, 1 
    ,
\end{equation}
then all terms of order $n \geq n_\epsilon $ are $O(\epsilon)$ in the series of formula \ref{formula:sym_c/n}, and that the error in the option price itself is bounded by
\begin{equation}
    \frac{ \alpha e^{(\alpha\delta - r)\tau}}{\sqrt{\pi}} \, \epsilon
\end{equation}
after the computation of $n_\epsilon + 1$ terms.

\paragraph{Asset-or-nothing} Recall the notations introduced in the proof of formula \ref{formula:sym_A/N} for the general term of the series:
\begin{equation}
        R_{n_1,n_2} \, := \,   \frac{ k_0^{n_1} }{n_1! \Gamma( 1+\frac{-n_1+n_2}{2} )} \,
        \K_{\frac{n_1-n_2+1}{2}}(\alpha\delta\tau)
        \left(  \frac{\delta\tau}{2\alpha} \right)^{\frac{-n_1+n_2+1}{2}}    
\end{equation}
and for the terms on the line $n_2=0$:
\begin{equation}
        R_{n_1} \, := \, \frac{ k_0^{n_1} }{n_1! \Gamma( 1 - \frac{n_1}{2} )} \,
        \K_{\frac{n_1+1}{2}}(\alpha\delta\tau)
        \left(  \frac{\delta\tau}{2\alpha} \right)^{\frac{-n_1+1}{2}}  
        .
\end{equation}  
Let us fix $n_1\in\mathbb{N}$ and consider
\begin{equation}\label{R_n1_n2_ratio}
    \left\vert \frac{R_{n_1,n_2+1}}{R_{n_1,n_2}} \right\vert
    \, = \, 
    \left\vert  \frac{\Gamma(1+\frac{-n_1+n_2}{2})}{\Gamma(1+\frac{-n_1+n_2+1}{2})}  \right\vert
    \,
    \left\vert \frac{ \K_{\frac{n_1-n_2}{2} }(\alpha\delta\tau)}{\K_{\frac{n_1-n_2+1}{2}}(\alpha\delta\tau)}   \right\vert
    \, 
    \sqrt{\frac{\delta\tau}{2\alpha}}
    .
\end{equation}
From the particular values of the Gamma functions \eqref{gamma_half_integers}, the ratio of Gamma functions in \eqref{R_n1_n2_ratio} is smaller or equal to $\sqrt{\pi}$, and the ratio of Bessel functions is smaller than 1, as a consequence of the symmetry and monotonicity relations \eqref{Bessel_sym} and \eqref{Bessel_monotonous}. Hence,
\begin{equation}
    \left\vert \frac{R_{n_1,n_2+1}}{R_{n_1,n_2}} \right\vert
    \, < \, \sqrt{\frac{\pi\delta\tau}{2\alpha}}
\end{equation}
and, consequently, $ \vert R_{n_1,n_2}  \vert < \vert R_{n_1,0} \vert = \vert R_{n_1} \vert $ for any $n_2$ in $\mathbb{N}$ as soon as
\begin{equation}\label{tau_bound}
    \tau \, < \, \frac{2\alpha}{\pi\delta}
    .
\end{equation}
Under this condition, all $R_{n_1,n_2}$ terms are therefore $O(\epsilon)$ as soon as $n_1,n_2 \geq n_\epsilon$ where $n_\epsilon$ is the one determined in \eqref{n_epsilon}, and, consequently, the error in the option price given formula \ref{formula:sym_A/N} is bounded by
\begin{equation}
    \frac{ K \alpha e^{(\alpha\delta - r)\tau}}{\sqrt{\pi}} \, \epsilon
\end{equation}
after the computation of $(n_\epsilon + 1)^2$ terms. Note that if \eqref{tau_bound} is not satisfied, the series still converges but the maximum is not attained on the line $n_2=0$, which complicates the estimation of the number of terms to compute. We may nevertheless observe that \eqref{tau_bound} is a very reasonable condition: for instance, using the implied parameters given in table \ref{tab:accessible} for SX5E options, we find $\tau < 9. 49$, which is very close to the maximal expiry (10 years) quoted for options written on this underlying.

\paragraph{European} Exactly the same analysis can be performed on the European option, resulting in an error for the option price given by formula \ref{formula:sym_eur} bounded by
\begin{equation}
    \frac{ K \alpha e^{(\alpha\delta - r)\tau}}{\sqrt{\pi}} \, \epsilon
\end{equation}
after the computation of $n_\epsilon (n_\epsilon +1)$ terms (because the $n_2$ summation starts at $n_2=1$). To illustrate these observations, we summarize in table \ref{tab:speed} the minimal rank, number of terms and price errors obtained for the digital and European options for some realistic market parameters.

\begin{table}[ht]
 \caption{ Rank $n_\epsilon$ beyond which the series terms in formulas \ref{formula:sym_A/N}, \ref{formula:sym_eur} and \ref{formula:sym_c/n} are $0(\epsilon)$, and corresponding truncation error on the option prices. Parameters: $S_t$=3800, $K=4000$, $r = 1\%$, $q = 0\%$, $\tau = 1$, $\alpha=8.9932$, $\delta = 1.1528$.}
 \label{tab:speed}       
 \centering
 \begin{tabular}{cccc}
 \hline
  \multicolumn{4}{c}{{\bfseries Asset-or-nothing} (Formula \ref{formula:sym_A/N}) }  \\
 \hline
 $\epsilon$  & Minimal rank $n_\epsilon = 2 p_\epsilon + 1$  & Number of terms $(n_\epsilon +1)^2$  & Price error   \\
 $10^{-5}$  & 5 & 36  & 6.39064  \\
 $10^{-10}$  & 11 & 144  & 0.0639064  \\
 $10^{-15}$  & 15 & 256 &  $6.39064 \times 10^{-7}$ \\
 $10^{-20}$  & 21  & 484  & $6.39064 \times 10^{-12}$ \\
 \hline
  \multicolumn{4}{c}{{\bfseries European} (Formula \ref{formula:sym_eur}) }  \\
 \hline
 $\epsilon$  & Minimal rank $n_\epsilon = 2 p_\epsilon + 1$  & Number of terms $n_\epsilon(n_\epsilon +1)$  & Price error   \\
 $10^{-5}$  & 5 & 30  & 6.39064  \\
 $10^{-10}$  & 11 & 132  & 0.0639064  \\
 $10^{-15}$  & 15 & 240 &  $6.39064 \times 10^{-7}$ \\
 $10^{-20}$  & 21  & 462  & $6.39064 \times 10^{-12}$ \\
 \hline
  \multicolumn{4}{c}{{\bfseries Cash-or-nothing} (Formula \ref{formula:sym_c/n})}  \\
  \hline
 $\epsilon$  & Minimal rank $n_\epsilon = 2 p_\epsilon + 1$  & Number of terms $n_\epsilon +1$  & Price error   \\
 $10^{-5}$  & 5 &  6 & 0.00159766  \\
 $10^{-10}$  & 11 &  12 & 0.0000159766  \\
 $10^{-15}$  & 15 &  16 &  $1.59766 \times 10^{-10}$ \\
 $10^{-20}$  & 21  & 22  & $1.59766 \times 10^{-15}$\\
 \hline
\end{tabular}
\end{table}

\subsection{Comparisons with Fourier techniques}

\paragraph{Lewis formula}
We recall that, following \cite{Lewis01}, digital option prices admit convenient representation involving the risk-neutral characteristic function and the log-forward moneyness; the asset-or-nothing call can be written as
\begin{equation}\label{Lewis_A/N}
    C_{a/n} \, = \, S_t \, e^{-q\tau} \, \left( \, \frac{1}{2} \,+ \, \frac{1}{\pi}\int\limits_0^{\infty} \, Re \, \left[ \frac{e^{i u k} {\Psi_L}(u-i,\tau)}{iu} \right] \ud u \, \right)
    ,
\end{equation}
and the cash-or-nothing call as
\begin{equation}\label{Lewis_C/N}
    C_{c/n} \, = \, e^{-r\tau} \, \left( \, \frac{1}{2} \,+ \, \frac{1}{\pi}\int\limits_0^{\infty} \, Re \, \left[ \frac{e^{i u k} {\Psi_L}(u,\tau)}{iu} \right] \ud u \, \right)
    ,
\end{equation}
where, here, $k:=\log\frac{S_t}{K}+(r-q)\tau$, and where the characteristic function $\Psi(u,t) \, = \, e^{ t  \psi(u)}$ has been normalized by the martingale adjustment:
\begin{equation}
    {\Psi_L}(u,t) \, := \, e^{i u \omega t} \Psi(u,t) \, = \, 
    e^{i u \omega t  + i \mu u t  \,- \, \delta t
    \left( \sqrt{ \alpha^2 - (\beta + i u)^2 } - \sqrt{ \alpha^2 - \beta^2 }
    \right) }
    ,
\end{equation}
so that the martingale condition ${\Psi}_L (-i,t)=1$ holds true. In table \ref{tab:a/n}, we compare the asset-or-nothing prices obtained by an application of formula \ref{formula:sym_A/N} (truncated at $n_1=n_2=max$) and of formula \ref{formula:asym_A/N} (truncated at $n_1=n_2=n_3 = max$), with a numerical evaluation of the Lewis formula \eqref{Lewis_A/N} performed via a classical recursive algorithm on $[0,10^4]$. Same comparison is made in table \ref{tab:c/n} for the cash-or-nothing prices. We observe the excellent agreement between our analytical result and numerical ones, as well as the fast convergence of the series. The convergence is particularly accelerated in the ATM situation (for instance in the symmetric model, only 3 terms are needed to obtain a precision of $10^{-3}$ in the cash-or-nothing price). It is slightly slower for deep OTM options: this is because $k_0 \sim \log S_t$ when $S_t\rightarrow 0$, and therefore the positive powers of $k_0$ tend to slow down the overall convergence speed. Note also that the convergence is more rapid in the symmetric than in the asymmetric model, because we choose an implied parameter $ |\beta| > 1 $ complying with the calibrations in table \ref{tab:accessible}; if we had chosen $|\beta| < 1$, then the positive powers of $\beta$ would have accelerated the convergence of the asymmetric series.

\begin{table}[ht]
 \caption{Prices of asset-or-nothing call options, obtained by truncations of formulas \ref{formula:sym_A/N} and \ref{formula:asym_A/N}, and by a numerical evaluation of \eqref{Lewis_A/N}. Parameters: $K=4000$, $r = 1\%$, $q = 0\%$, $\tau = 1$, $\alpha=8.9932$, $\delta = 1.1528$.}
 \label{tab:a/n}       
 \centering
 \begin{tabular}{lccccc}
 \hline
 \multicolumn{6}{c}{{\bfseries Symmetric model} [$\beta = 0$] }  \\
 \hline
 & \multicolumn{4}{c}{Formula \ref{formula:sym_A/N}} & Lewis \eqref{Lewis_A/N} \\
 & $max=3$ & $max=5$ & $max=10$ & $max=15$ &   \\
 Deep OTM ($S_t=3000$)  & 861.9096  & 796.515 & 804.8118 &  804.9099 & 804.9097  \\
 OTM ($S_t=3500$)       &  1495.76986 & 1493.3213 & 1493.5276 & 1493.5278  & 1493.5278  \\
 ATM                   & 2309.8330  & 2313.6169 & 2313.7110 & 2313.7110  & 2313.7110  \\
 ITM ($S_t=4500$)        &  3163.3516 & 3170.7414 & 3170.9431 & 3170.9431  & 3170.9431   \\
 Deep ITM ($S_t=5000$)   &  3986.4269 & 3999.5086  & 3999.8854 & 3999.8852  & 3999.8852  \\
 \hline
\multicolumn{6}{c}{{\bfseries Asymmetric model} [$\beta = -4.5176$] }  \\
 \hline
 & \multicolumn{4}{c}{Formula \ref{formula:asym_A/N}} & Lewis \eqref{Lewis_A/N} \\
 & $max=10$ & $max=20$ & $max=30$ & $max=50$ &   \\
 Deep OTM ($S_t=3000$) & 1084.9112 & 991.4964 & 990.8328 & 990.8302  & 990.8302  \\
 OTM ($S_t=3500$)      & 1814.0381  & 1705.6678 & 1704.8935 &  1704.8905  & 1704.8905 \\
 ATM                 &  2593.7092 & 2480.0154 &  2479.11828  &  2479.1149 & 2479.1149  \\
 ITM ($S_t=4500$)      & 3310.5927 & 3252.0495 & 3250.4093 & 3250.4089  &  3250.4089 \\
 Deep ITM ($S_t=5000$) &  3777.9899 & 4003.6194 & 3989.4277 & 3989.7291 & 3989.7293  \\
 \hline
\end{tabular}
\end{table}

\begin{table}[ht]
 \caption{Prices of cash-or-nothing call options, obtained by truncations of formulas \ref{formula:sym_c/n} and \ref{formula:asym_c/n}, and by a numerical evaluation of \eqref{Lewis_A/N}. Parameters: $K=4000$, $r = 1\%$, $q = 0\%$, $\tau = 2$, $\alpha=8.9932$, $\delta = 1.1528$.}
 \label{tab:c/n}       
 \centering
 \begin{tabular}{lccccc}
 \hline
 \multicolumn{6}{c}{{\bfseries Symmetric model} [$\beta = 0$] }  \\
 \hline
 & \multicolumn{4}{c}{Formula \ref{formula:sym_c/n}} & Lewis \eqref{Lewis_C/N} \\
 & $max=3$ & $max=5$ & $max=10$ & $max=15$ &   \\
 Deep OTM ($S_t=3000$)  & 0.2127  & 0.2092  &  0.2095  &  0.2095  &  0.2095  \\
 OTM ($S_t=3500$)       &  0.3076 & 0.3073  & 0.3073   &  0.3073  &  0.3073  \\
 ATM                   & 0.4054  & 0.4054  &  0.4054  &  0.4054  &  0.4054    \\
 ITM ($S_t=4500$)       &  0.4973 & 0.4973  &  0.4973  & 0.4973   &  0.4973    \\
 Deep ITM ($S_t=5000$)   & 0.5793  & 0.5793  & 0.5793   &  0.5793  &  0.5793   \\
 \hline
\multicolumn{6}{c}{{\bfseries Asymmetric model} [$\beta = -4.5176$] }  \\
 \hline
 & \multicolumn{4}{c}{Formula \ref{formula:asym_c/n}} & Lewis \eqref{Lewis_C/N} \\
 & $max=10$ & $max=20$ & $max=30$ & $max=50$ &   \\
 Deep OTM ($S_t=3000$)  & 0.2579  & 0.2360 & 0.2357  & 0.2357 & 0.2357   \\
 OTM ($S_t=3500$)       &  0.3523 &  0.3244 & 0.3240  &  0.3240 & 0.3240   \\
 ATM                  &  0.4544 &  0.4077 & 0.4074  & 0.4074  & 0.4074    \\
 ITM ($S_t=4500$)       & 0.5740   &  0.4823 & 0.4827  &  0.4827  &  0.4827  \\
 Deep ITM ($S_t=5000$)  & 0.7634  & 0.7277  &  0.5733 & 0.5452  &  0.5489   \\
 \hline
\end{tabular}
\end{table}

\paragraph{Carr-Madan formula}
Regarding European options, we recall the representation given in \cite{Carr99} based on the introduction of a dampling factor $a$ to avoid the divergence in $u=0$; namely, let
\begin{equation}
    \Psi_{CM} (u,t) \, := \, e^{i u [\log S_t + ( r - q + \omega ) t]} \, \Psi (u,t) 
    ,
\end{equation}
then the European call price admits the representation:
\begin{equation}\label{CarrMadan}
    C_{eur} \, = \, \frac{e^{-a\log K - r\tau}}{\pi} \, \int\limits_0^{\infty} \,
    e^{-i u \log K}
    Re \left[  \frac{\Psi_{CM} (u-(a+1)i,\tau)}{a^2+a-u^2+i(2a+1)u}  \right]
    \, \ud u
    ,
\end{equation}
where $a<0<a_{max}$, and $a_{max}$ is determined by the square integrability condition $\Psi_{CM} (-(a+1)i,\tau)<\infty$. In table \ref{tab:eur_Carr_Madan} we compare the European prices obtained by formula \ref{formula:sym_eur} (truncated at $n_1=n_2=max$) and formula \ref{formula:asym_eur} (truncated at $n_1=n_2=n_3=max$), with a numerical evaluation of the Carr-Madan formula \eqref{CarrMadan} on the interval $[0,10^4]$. We also observe the excellent agreement between our analytical results and the numerical ones, as well as the accelerated convergence for very short term options. For instance, when $\tau=$ 1 day, $(1+5)^2$ iterations are enough to obtain a precision of $10^{-3}$ in the option price in the symmetric model; this is because, when $S_t$ is close to $K$, then $k_0 \sim (r-q+\omega)\tau $ and therefore when $\tau\rightarrow 0$ the positive powers of $k_0$ arising in formulas \ref{formula:sym_eur} and \ref{formula:asym_eur} accelerate the convergence of the series. Note that, on the contrary, the short maturity case is not a favorable situation for a numerical evaluation of the Carr-Madan formula, because of the presence of oscillations of the integrand that considerably slow down the numerical Fourier inversion process.

\begin{table}[ht]
 \caption{Prices of European call options of various maturities, obtained by truncations of formulas \ref{formula:sym_eur} and \ref{formula:asym_eur}, and by a numerical evaluation of \eqref{CarrMadan}. Parameters: $S_t=K=4000$, $r = 1\%$, $q=0\%$, $\alpha=8.9932$, $\delta = 1.1528$.}
 \label{tab:eur_Carr_Madan}       
 \centering
 \begin{tabular}{lccccc}
 \hline
 \multicolumn{6}{c}{{\bfseries Symmetric model} [$\beta = 0$] }  \\
 \hline
 & \multicolumn{4}{c}{Formula \ref{formula:sym_eur}} & Carr-Madan \eqref{CarrMadan} \\
 Maturity & $max=3$ & $max=5$ & $max=10$ & $max=15$ &   \\
 1 year  & 576.6432  &  580.4319 &   580.5260 & 580.5260  & 580.5260   \\
 1 month &  150.8024 & 150.8651  &  150.8656 & 150.8656  & 150.8656  \\
 1 week   & 60.9649  &  60.9746 &  60.9747 & 60.9747  & 60.9747 \\
 1 day    & 15.4503  &   15.4515 &  15.4515  &  15.4515  & 15.4515  \\
 \hline
\multicolumn{6}{c}{{\bfseries Asymmetric model} [$\beta = -4.5176$] }  \\
 \hline
 & \multicolumn{4}{c}{Formula \ref{formula:asym_eur}} & Carr-Madan \eqref{CarrMadan} \\
 Maturity & $max=10$ & $max=20$ & $max=30$ & $max=50$ &   \\
 1 year  & 790.330   & 679.6635  & 678.8152  & 678.8118  &  678.8118  \\
 1 month  & 173.6275   & 173.5547  & 173.5546  &  173.5546 &  173.5546  \\
 1 week   &  68.4327  &  68.4234 & 68.4234  & 68.4234  & 68.4234  \\
 1 day     &  16.7801  &  16.7790 &  16.7790 & 16.7790  & 16.7790  \\
 \hline
\end{tabular}
\end{table}

\paragraph{Fast Fourier Transform}
The Fast Fourier Transform (FFT) is an algorithm allowing to compute efficiently sums of the type
\begin{equation}\label{FFT}
    \sum_{j=1}^N e^{-i\frac{2\pi}{N}(j-1)(v-1)}x(j) \, , \hspace*{0.5cm} N\in\mathbb{N}
    .
\end{equation}
Such an algorithm can be successfully applied to provide numerical approximations of \eqref{CarrMadan}. Let $\kappa:=\log K$; following \cite{Carr99} we truncate the integral in \eqref{CarrMadan} and evaluate it by a trapezoidal rule for $N$ equally spaced panels of length $\eta$ (i.e., the upper bound in \eqref{CarrMadan} is truncated at $N\eta$); for $v=1,2,\dots N$, we introduce the collection of log-strikes
\begin{equation}
    \kappa_v \, := \, \frac{\lambda N}{2} \, + \, \lambda(v-1)
    ,
\end{equation}
generating $N$ log-strike prices located in the interval $[-\frac{\lambda N}{2} , \frac{\lambda N}{2}-\lambda]$. Choosing $\lambda=\frac{2\pi}{N\eta}$ allows to re-write \eqref{CarrMadan} as:
\begin{equation}\label{CarrMadan_Trapezoidal}
    C_{eur} \, \simeq \, \frac{e^{-a\kappa_v - r\tau}}{\pi} \, 
    \sum\limits_{j=1}^N
    Re\left[
    e^{-i\frac{2\pi}{N}(j-1)(v-1)}e^{i\frac{\lambda N}{2}(j-1)\eta} 
    x(j)
    \eta
    \right]
    ,
\end{equation}
where 
\begin{equation}
    x(j) \, := \, 
    \frac{\Psi_{CM}((j-1)\eta,\tau)}{a^2+a-((j-1)\eta)^2+i(2a+1)(j-1)\eta}
    .
\end{equation}
We can observe that \eqref{CarrMadan_Trapezoidal} is of the form \eqref{FFT}; using Simpson's rule weightings, \cite{Carr99} obtain the following approximation for the European call price:
\begin{equation}\label{CarrMadan_FFT}
    C_{eur} \, \simeq \, \frac{e^{-a\kappa_v - r\tau}}{\pi} \, 
    \sum\limits_{j=1}^N
    Re\left[
    e^{-i\frac{2\pi}{N}(j-1)(v-1)}e^{i\frac{\lambda N}{2}(j-1)\eta} 
    x(j)
    \frac{\eta}{3}
    (3+(-1)^j-\delta_{j-1})
    \right]
\end{equation}
where $\delta_x$ denotes the Kronecker symbol. In figure \ref{fig:FFT}, we fix $S_t=3000$, $\alpha=40$, $\delta=25$, $\tau=1/12$ (one month expiry), $r=1\%$, and we compare the results obtained by the FFT algorithm \eqref{CarrMadan_FFT} and applications of the pricing formula \ref{formula:sym_eur} for the European call in the symmetric model:
\begin{itemize}
    \item [-] Formula \ref{formula:sym_eur} is truncated at $n_1=n_2=30$; there are, therefore $31\times 30 = 930$ terms to compute but, as discussed in subsection \ref{subsec:practical}, $14\times 15 = 210$ terms are actually equal to zero, and therefore there are only $720$ non null terms to compute;
    \item[-]  In the FFT algorithm we follow \cite{Carr99} and choose $\eta=0.25$ for the spacing parameter, and we choose $N=500$ (left graph) or $N=1000$ (right graph) for the truncation parameter.
\end{itemize}
We observe that, in both cases, the pricing formula \ref{formula:sym_eur} and the FFT algorithm display excellent agreement. However, in the case $N=500$, only $18$ strikes are attainable in the interval $[2000,5000]$ via the FFT method, and 36 when $N=1000$; in the first (resp. second) case, strikes prices are separated by 100 to 200 (resp. 50 to 100) points. To get a collection of strikes separated by only 10 points (at least when one is not too far from the money), one would need to choose $N=5000$, and even $N=50 \, 000$ to get consecutive prices. This is to be compared with the $720$ terms needed by the pricing formula \ref{formula:sym_eur} to provide a continuum of strikes across the whole interval.

\begin{figure}[H]
\centering
\includegraphics[scale=0.6]{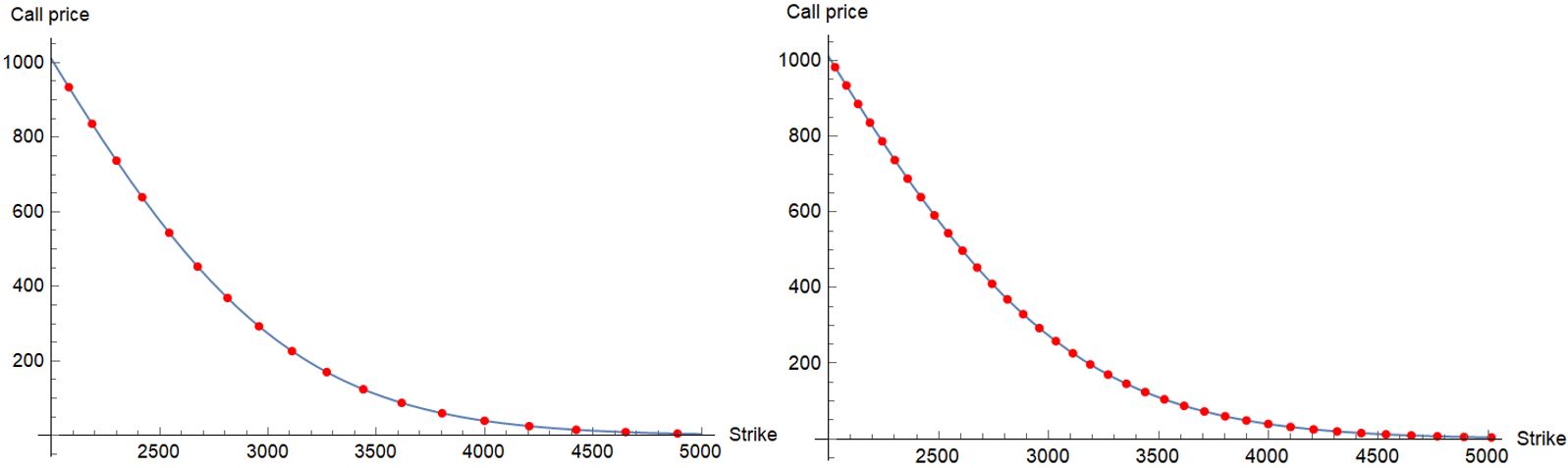}
\caption{Blue curve: pricing formula \ref{formula:sym_eur} truncated at $n_1=n_2=30$; red dots: FFT algorithm \eqref{CarrMadan_FFT} with truncation parameter $N$ and spacing parameter $\eta=0.25$. In the left graph, $N=500$, and in the right graph, $N=1000$.}
\label{fig:FFT}
\end{figure}

\subsection{Comparisons with Monte Carlo simulations}
Let $n\in\mathbb{N}\backslash \{0\}$ and define the family of independent and identically distributed random variables $Z^{(i)}$, $i=1\dots n$, all distributed according to the symmetric NIG distribution $ Z^{(i)} \, \sim \, \mathrm{NIG} (\alpha, 0 ,\delta, \mu)$, and define
\begin{equation}
   C_{log}^{(i)} \, := \, e^{-r\tau} \, \left[ \log \left( \frac{S_t}{K}e^{(r-q+\omega)\tau + Z^{(i)}} \right)   \right]^+ \, = \, e^{-r\tau} \, [k_0 + Z^{(i)}]^+
\end{equation}
as well as
\begin{equation}\label{MC_LogCall}
    C_{log}^{(n)} \, := \, \frac{1}{n} \, \sum\limits_{i=1}^n C_{log}^{(i)}
    .
\end{equation}
We know from the strong law of large numbers that $C_{log}^{(n)}$ converges to the price of the log call option, more precisely that
\begin{equation}
    C_{log}^{(n)}\,  \longrightarrow \, \mathbb{E}^{\mathbb{Q}} \left[ e^{-r\tau} \left[ \log\frac{S_T}{K} \right]^+ \, \vert \, S_t \right]
\end{equation}
almost surely when $n\rightarrow\infty$. Similarly, regarding power options, we define (in the European case):
\begin{equation}\label{MC_PowerCall}
     C_{pow}^{(i)} \, := \, e^{-r\tau} \,
    \left[ S^a e^{a ( (r - q +\omega)\tau + Z^{(i)} ) } \, - \, K    \right]^+  
    \, , \hspace{0.5cm}
      C_{pow}^{(n)} \, := \, \frac{1}{n} \sum\limits_{i=1}^n C_{pow}^{(i)}
\end{equation}
and, for the capped digital option,
\begin{equation}\label{MC_CappedCall}
     C_{capped\,c/n}^{(i)} \, := \, e^{-r\tau} \,
    \mathbbm{1}_{ \{ -k_{0,-} < Z^{(i)} < -k_{0,+}  \} }  
    \, , \hspace{0.5cm}
      C_{capped \, c/n}^{(n)} \, := \, \frac{1}{n} \sum\limits_{i=1}^n C_{capped \, c/n}^{(i)}
\end{equation}
which converge to the European power call and to the capped cash-or-nothing call respectively. In table \ref{tab:log_power_MC}, we compare the results obtained via the Monte Carlo simulations \eqref{MC_LogCall}, \eqref{MC_PowerCall} and \eqref{MC_CappedCall} for different number of paths, with truncations of the pricing formulas \ref{form:log}, \ref{form:pow} and of \eqref{sym_cap_c/n}. As expected, the results display good agreement, but our series provide a far more precise price and a far more rapid convergence: for instance, only 2 to 4 terms are needed to obtain a level of precision of $10^{-3}$ for the log call using formula \ref{form:log}, while the Monte Carlo price still features a relative error of 1\% in the OTM case and even 4\% in the ITM case. Note also that, defining the  95\% confidence interval by $C_{log}^{(n)} \pm 1.96 \, \sigma_P / \sqrt{n}$ where
\begin{equation}
    \sigma_P \, := \, \sqrt{\mathrm{var} \{ C_{log}^{(i)} \}_{i=1 \dots n}  }
    ,
\end{equation}
then its length vary between $0.0136$ (OTM case) and $0.0187$ (ITM case) after $n=1000$ paths. Of course the confidence interval could be reduced by increasing the number of paths (but then the Standard Monte Carlo becomes time and resource consuming) or by introducing variance reduction techniques, such as antithetic variates or importance sampling methods (see \cite{Su00} or the classical monograph \cite{Glasserman04}). On the contrary, with our series expansions, the results are quasi instantaneous and can easily be made as precise as one wishes, without introducing further sophistication.

\begin{table}[ht]
 \caption{Prices of log, power and capped calls, obtained by Monte Carlo simulations ($n$ paths) or truncation of formulas \ref{form:log}, \ref{form:pow} and series \eqref{sym_cap_c/n}. Parameters: $K_-=K =4000$, $K_+=5000$, $r = 1\%$, $q = 0\%$, $\tau = 2$, $\alpha=8.9932$, $\delta = 1.1528$, $a=1.2$.}
 \label{tab:log_power_MC}       
 \centering
 \begin{tabular}{lcccccc}
 \hline
 \multicolumn{7}{c}{{\bfseries Log option (call) } }  \\
 \hline
 & \multicolumn{3}{c}{ Monte Carlo \eqref{MC_LogCall} } & \multicolumn{3}{c}{ Formula \ref{form:log} } \\
 &  $n=100$  &  $n=500$ & $n=1000$  & $n_{max} = 1$  &  $n_{max} = 3$  &  $n_{max} = 5$\\ 
 OTM ($S_t=3500)$  & 0.0826  & 0.1034 & 0.1002 & 0.1012 & 0.1008 & 0.1008 \\
 ATM ($S_t=4000)$  & 0.1537  & 0.1508 & 0.1509 & 0.1483 & 0.1482 & 0.1482 \\
 ITM ($S_t=4500)$  & 0.2428  & 0.2255 & 0.1923 & 0.2014 & 0.2014 & 0.2014 \\
 \hline
  \multicolumn{7}{c}{{\bfseries Power option (European) } }  \\
 \hline
 & \multicolumn{3}{c}{ Monte Carlo \eqref{MC_PowerCall} } & \multicolumn{3}{c}{ Formula \ref{form:pow} } \\
 &  $n=100$  &  $n=1000$ & $n=5000$  & $max = 20$  &  $max = 40$  &  $max = 60$\\ 
 OTM ($S_t=3500)$  &  12943.90  & 13976.71  & 14456.01  & 1429.53  & 14629.84  & 14629.84  \\
 ATM ($S_t=4000)$  &  17229.06  & 17263.31  & 17678.74  & 17843.79  &  17847.18 & 17847.18  \\
 ITM ($S_t=4500)$  &  20719.09  & 20310.75  & 21422.76  & 21126.01  & 21148.88  & 21148.89  \\
 \hline
 \multicolumn{7}{c}{{\bfseries Capped option (digital)} }  \\
 \hline
 & \multicolumn{3}{c}{ Monte Carlo \eqref{MC_CappedCall} } & \multicolumn{3}{c}{ Series \eqref{sym_cap_c/n} } \\
 &  $n=100$  &  $n=1000$ & $n=5000$  & $n_{max} = 1$  &  $n_{max} = 5$  &  $n_{max} = 10$\\ 
 OTM ($S_t=3500)$  & 0.1764  & 0.1519  & 0.1262  & 0.1754  & 0.1355  & 0.1347 \\
 ATM ($S_t=4000)$  & 0.1862  & 0.1608  & 0.1598  & 0.1754  & 0.1575  & 0.1575 \\
 ITM ($S_t=4500)$  & 0.2058  & 0.1774  & 0.1672  & 0.1754  & 0.1702  & 0.1702 \\
 \hline
\end{tabular}
\end{table}

\section{Concluding remarks}

In this paper, we have proved two general formulas for pricing arbitrary path independent instruments in the exponential NIG model, in the symmetric and asymmetric cases. These formulas allow to express the Mellin transform of the instrument's price as the product of the Mellin transform of the instrument's payoff and of the NIG probability density. Inverting the formulas by means of residue theory in $\mathbb{C}$ and $\mathbb{C}^n$ has allowed us to derive practical closed-form pricing formulas for various path independent options and contracts, under the form of quickly convergent series. The convergence of the series is guaranteed as soon as a simple condition of the log forward moneyness and on the option's maturity is fulfilled. We have tested our results by comparing them with classical numerical methods, and provided precise estimate for the convergence speed; notable feature is that a very reasonable number of terms is required to obtain an excellent level of precision, and that the convergence is particularly fast for short term and at the money options.

Future work should include, among others, an extension of the Mellin residue summation method to path independent instruments on several assets, and to path dependent instruments. Asian options with continuous geometric payoffs, in particular, should be investigated, because the characteristic function for the geometric average is known exactly in the exponential NIG model (see \cite{Fusai08}), for both fixed and floating strikes.

Extension of the technique to Generalized Hyperbolic (GH) L\'evy motions (see \cite{Prause99,Eberlein01}) should also be considered; indeed, the probability density of the GH distribution has a very similar form to the NIG density \eqref{NIG_distribution_density}, which, at first sight, allows for the same convenient representation in terms of Mellin-Barnes integrals for the Bessel kernel. However, GH distributions are not convolution-closed, that is, the L\'evy processes they generate are not necessarily distributed according to a GH distribution for increments of length $t\neq 1$ (exceptions being the NIG process, which, as we know, is distributed according to a NIG distribution $\mathrm{NIG}(\alpha,\beta,\delta t,\mu t)$ for all $t$, as well as the generalized Laplace or Variance Gamma distribution). As a consequence, a Mellin-Barnes integral representation for the density of the GH process is not straightforward to derive, but could nevertheless be obtained from the moment generating function, after suitable transformations from the Laplace space to the Mellin space.

\section*{Acknowledgments}

The author thanks Ryan McCrickerd for insightful comments and discussions. The author also thanks two anonymous Reviewers and the Managing Editor for their careful reading of the manuscript, and their valuable remarks and suggestions.

\singlespacing

\doublespacing

\appendix

\section{Brief review of the Mellin transform}\label{app:Mellin}

We present an overview of the one-dimensional Mellin transform; this theory is explained in full detail in \cite{Flajolet95}, and table of Mellin transforms can be found in any monograph on integral transforms (see e.g. \cite{Bateman54}).


\noindent 1. The Mellin transform of a locally continuous function $f$ defined on $\mathbb{R}^+$ is the function $f^*$ defined by
\begin{equation}\label{Mellin_def}
    f^*(s) \, := \, \int\limits_0^\infty \, f(x) \, x^{s-1} \, \ud x
    .
\end{equation}
The region of convergence $\{ \alpha < Re (s) < \beta \}$ into which the integral \eqref{Mellin_def} converges is often called the fundamental strip of the transform, and sometimes denoted $ < \alpha , \beta  > $.

\noindent 2. The Mellin transform of the exponential function is, by definition, the Euler Gamma function:
\begin{equation}\label{Gamma_def}
    \Gamma(s) \, = \, \int\limits_0^\infty \, e^{-x} \, x^{s-1} \, \ud x
\end{equation}
with strip of convergence $\{ Re(s) > 0 \}$. Outside of this strip, it can be analytically continued, except at every negative $s=-n$ integer where it admits the singular behavior
\begin{equation}\label{sing_Gamma}
    \Gamma(s) \, \underset{s\rightarrow -n}{\sim} \, \frac{(-1)^n}{n!}\frac{1}{s+n} \, ,
    \hspace*{0.5cm} n\in\mathbb{N}
    .
\end{equation}
In table \ref{tab:Mellin} we summarize the main Mellin transforms used in this paper, as well as their convergence strips.
\begin{table}[ht]
 \caption{Mellin pairs used throughout the paper.} 
 \label{tab:Mellin}       
 \centering
 \begin{tabular}{|c|c|c|}
 \hline
 $f(x)$  &  $f^*(s)$  &  Convergence strip  \\
 \hline   
 $e^{-a x}$       &  $a^{-s} \Gamma(s)$  &  $< 0 , \infty >$  \\
 $e^{-a x} - 1 $  &  $a^{-s} \Gamma(s)$  &  $< -1 , 0 >$  \\
 $\K_{\nu}(ax)$    & $a^{-s}2^{s-2} \Gamma\left( \frac{s-\nu}{2} \right)\Gamma\left( \frac{s+\nu}{2} \right)$  &   $ < |Re(\nu)|, \infty > $ \\
 $\frac{\K_\nu(a\sqrt{x^2+b^2})}{(x^2+b^2)^{\frac{\nu}{2}}}$  & $a^{\frac{s}{2}}2^{\frac{s}{2}-1}b^{\frac{s}{2}-\nu}\Gamma(\frac{s}{2})\K_{\nu-\frac{s}{2}}(ab)$  & $< 0 , \infty >$   \\
 \hline
\end{tabular}
\end{table} 

\noindent 3. The inversion of the Mellin transform is performed via an integral along any vertical line in the strip of convergence:
\begin{equation}\label{inversion}
    f(x) \, = \, \int\limits_{c-i\infty}^{c+i\infty} \, f^*(s) \, x^{-s} \, \frac{\ud s}{2i\pi} \hspace*{1cm} c\in ( \alpha, \beta )
\end{equation}
and notably for the exponential function one gets the so-called \textit{Cahen-Mellin integral}:
\begin{equation}\label{Cahen}
    e^{-x} \, = \, \int\limits_{c-i\infty}^{c+i\infty} \, \Gamma(s) \, x^{-s} \, \frac{\ud s}{2i\pi}, \hspace*{0.5cm} c>0
    .
\end{equation}

\noindent 4. When $f^*(s)$ is a ratio of products of Gamma functions of linear arguments:
\begin{equation}
    f^*(s) \, = \, \frac{\Gamma(a_1 s + b_1) \dots \Gamma(a_m s + b_m)}{\Gamma(c_1 s + d_1) \dots \Gamma(c_l s + d_l)}
\end{equation}
then one speaks of a \textit{Mellin-Barnes integral}, whose \textit{characteristic quantity} is defined to be
\begin{equation}
    \Delta \, = \, \sum\limits_{k=1}^m \, a_k \, - \, \sum\limits_{j=1}^l \, c_j
    .
\end{equation}
$\Delta$ governs the behavior of $f^*(s)$ when $|s|\rightarrow \infty$ and thus the possibility of computing \eqref{inversion} by summing the residues of the analytic continuation of $f^*(s)$ right or left of the convergence strip:
\begin{equation}
    \left\{
    \begin{aligned}
        & \Delta < 0 \hspace*{1cm} f(x) \, = \, -\sum\limits_{Re(s) > \beta} \, \res \left[f^*(s)x^{-s}\right],  \\
        & \Delta > 0 \hspace*{1cm} f(x) \, = \, \sum\limits_{Re(s) < \alpha} \, \res \left[f^*(s)x^{-s}\right]
        .
    \end{aligned}
    \right.
\end{equation}
For instance, in the case of the Cahen-Mellin integral one has $\Delta = 1$ and therefore:
\begin{equation}
e^{-x} \, = \, \sum\limits_{Re(s)<0} \res \left[ \Gamma(s) \, x^{-s} \right] \,  =  \, \sum\limits_{n=0}^{\infty} \, \frac{(-1)^n}{n!}x^n
\end{equation}
as expected from the usual Taylor series of the exponential function.

\section{Some useful special functions identities}\label{app:special}

We list some properties of special functions that are used throughout the paper; more details can be found e.g. in \cite{Abramowitz72,Andrews92}.

\subsection{Gamma function}
\paragraph{Particular values} The Gamma function $\Gamma(s)$ has been defined in \eqref{Gamma_def} for $Re(s)>0$; integrating by parts shows that it satisfies the functional relation $\Gamma(s+1) = s\Gamma(s)$; as $\Gamma(1)=1$, it follows that 
\begin{equation}
    \Gamma(n+1) = n! \, , \hspace{0.3cm} n\in\mathbb{N}
\end{equation}
and that the analytic continuation of $\Gamma(s)$ to the negative half-plane is singular at every negative integer $-n$ with residue $\frac{(-1)^n}{n!}$. Other useful identities include $\Gamma(\frac{1}{2}) = \sqrt{\pi}$ and, more generally,
\begin{equation}\label{gamma_half_integers}
    \left\{
    \begin{aligned}
        & \Gamma\left( \frac{1}{2} - n \right) \, = \, \frac{(-1)^n 4^n n!}{(2n)!} \, \sqrt{\pi} 
        \\
        & \Gamma\left( \frac{1}{2} + n \right) \, = \, \frac{(2n)!}{4^n n!} \, \sqrt{\pi}
        .
    \end{aligned}
    \right.
\end{equation}
for $n\in\mathbb{N}$.

\paragraph{Stirling approximation} We recall the well-known Stirling approximation for the factorial:
\begin{equation}\label{Stirling}
    n! \, \underset{n\rightarrow\infty}{\sim} \, \sqrt{2\pi n} \, n^n \, e^{-n}
    .
\end{equation}

\paragraph{Legendre duplication formula} For any $s \in \mathbb{C}$, we have:
\begin{equation}\label{Legendre}
    \frac{\Gamma\left( \frac{s}{2} \right)}{\Gamma(s)} \, = \, \frac{\sqrt{\pi}}{2^{s-1}} \, \frac{1}{\Gamma\left( \frac{s+1}{2} \right)}
    .
\end{equation}

\paragraph{Pochhammer symbol} The Pochhamer symbol $(a)_n$, sometimes denoted by the Appel symbol $(a,n)$, and also called rising factorial, is defined by
\begin{equation}\label{Pochhammer_def}
    (a)_n \, := \, \frac{\Gamma(a+n)}{\Gamma(a)} \, , \hspace{0.3cm} a \notin \mathbb{Z}_-
    .
\end{equation}
The definition \eqref{Pochhammer_def} extends continuously to negative integers thans to the functional relation $\Gamma(s+1) = s \Gamma(s)$, thanks to the relation:
\begin{equation}\label{Pochhammer_def_neg}
    (-k)_n \, = \,
    \left\{
    \begin{aligned}
        & \frac{(-1)^n k!}{(k-n)!} & 0\leq n <k\\
        & 0  &  n>k    
        .
    \end{aligned}    
    \right.
\end{equation}
where $k\in\mathbb{N}$.

\subsection{Bessel functions}

The modified Bessel function of the second kind, also called MacDonald function, can be defined by the Mellin integral
\begin{equation}\label{Bessel_def}
    \K_\nu(z) \, := \, \frac{1}{2} \, \left(\frac{z}{2}\right)^\nu \,  \int\limits_0^{\infty} \, 
    e^{-t - \frac{z^2}{4t} } \, t^{-\nu - 1} \, \ud t
\end{equation}
for $|\mathrm{arg} z| < \frac{\pi}{4}$. It follows that $\K_\nu(z)$ has the symmetry property:
\begin{equation}\label{Bessel_sym}
    \K_{\nu} (z) \, = \, \K_{-\nu} (z)
\end{equation}
and has monotonous absolute values:
\begin{equation}\label{Bessel_monotonous}
    0 \leq \nu_1 < \nu_2 \,\, \Longrightarrow \,\, \vert \K_{\nu_1} (z) \vert \, <\, \vert \K_{\nu_2} (z) \vert
    .
\end{equation}

\paragraph{Large index} When $\nu\rightarrow \infty$, one has the following behavior:
\begin{equation}\label{Bessel_large_index}
    \K_\nu(z) \, \underset{\nu\rightarrow\infty}{\sim} \, 
    \sqrt{\frac{\pi}{2\nu}} \, \left( \frac{ez}{2\nu} \right)^{-\nu}
    .
\end{equation}

\paragraph{Large argument (Hankel's expansion)} Define the following sequence:
\begin{equation}\label{Hankel_a}
    \left\{
    \begin{aligned}
        & a_0(\nu) \, = \, 1 \\
        & a_k(\nu) \, = \, \frac{(4\nu^2 - 1^2)(4\nu^2-3^2) \dots (4\nu^2 - (2k-1)^2)}{k! 8^k} \, , \hspace{0.3cm} k\geq 1
        .
    \end{aligned}
    \right.
\end{equation}
Then, for large $z$ and fixed $\nu$, we  have:
\begin{align}\label{Bessel_large_z}
    \K_\nu(z) & \underset{z\rightarrow\infty}{=} \sqrt{\frac{\pi}{2z}} \, e^{-z} \, \sum\limits_{k=0}^{\infty} \, \frac{a_k(\nu)}{z^k}
            .
\end{align}
In particular, when $4\nu^2-1 = 0$, i.e. when $\nu=\frac{1}{2}$, all the $a_k(\nu)$ are null in definition \eqref{Hankel_a} when $k\geq 1$, and we are left with:
\begin{equation}\label{Bessel_1/2}
    \K_{\frac{1}{2}}(z) \, = \, \sqrt{\frac{\pi}{2z}} \, e^{-z}
\end{equation}
for all z.




\begin{thebibliography}{99}

\bibitem[Abramowitz \& Stegun(1972)]{Abramowitz72}
Abramowitz, M. and Stegun, I., Handbook of Mathematical Functions, Dover Publications, Mineola, NY (1972)

\bibitem[Aguilar(2019)]{Aguilar19}
Aguilar, J. Ph., On expansions for the Black-Scholes prices and hedge parameters, Journal of Mathematical Analysis and Applications {\bfseries 478(2)}, 973-989 (2019)

\bibitem[Aguilar \& Korbel(2019)]{AK19}
Aguilar, J. Ph. and Korbel, J.,  Simple Formulas for Pricing and Hedging European Options in the Finite Moment Log-Stable Model, Risks {\bfseries 7}, 36 (2019)

\bibitem[Aguilar(2020)]{Aguilar20}
Aguilar, J. Ph., Some pricing tools for the Variance Gamma model, International Journal of Theoretical and Applied Finance {\bfseries 23(4)}, 2050025 (2020)

\bibitem[Albrecher \& Predota(2004)]{Albrecher04}
Albrecher, H. and Predota, M., On Asian option pricing for NIG L\'evy processes, Journal of Computational and Applied Mathematics {\bfseries 172}, 153-168 (2004)

\bibitem[Albrecher \& Schoutens(2005)]{Albrecher05}
Albrecher, H. and Schoutens, W., Static hedging of Asian options under stochastic volatility models using Fast Fourier Transform. In: A. Kyprianou et al. (Eds), Exotic Options and Advanced L\'evy models pp. 129-148, John Wiley \& Sons, Hoboken, NJ (2005)

\bibitem[Andrews(1992)]{Andrews92}
Andrews, L.C., Special Functions of Mathematics for Engineers, McGraw-Hill Book Company, New York (1992)

\bibitem[Barndorff-Nielsen(1977)]{Barndorff77}
Barndorff-Nielsen, O., Exponentially decreasing distributions for the logarithm of particle size, Proceedings of the Royal Society of London {\bfseries 353}, 401-419 (1977)

\bibitem[Barndorff-Nielsen(1995)]{Barndorff95}
Barndorff-Nielsen, O., Normal inverse Gaussian distributions and the modeling of stock returns, Research report no 300, Department of Theoretical Statistics, Aarhus University (1995)

\bibitem[Barndorff-Nielsen(1997)]{Barndorff97}
Barndorff-Nielsen, O., Normal inverse Gaussian distributions and stochastic volatility models, Scandinavian Journal of Statistics {\bfseries 24(1)}, 1-133 (1997)

\bibitem[Bateman(1954)]{Bateman54}
Bateman, H., Tables of Integral Transforms (vol. I and II), McGraw-Hill Book Company, New York (1954)

\bibitem[Bertoin(1996)]{Bertoin96}
Bertoin, J., L\'evy Processes, Cambridge University Press, Cambridge, New York, Melbourne (1996)

\bibitem[Black \& Scholes(1973)]{Black73}
Black, F. and Scholes, M., The Pricing of Options and Corporate Liabilities, Journal of Political Economy {\bfseries 81(3)}, 637-654 (1973)

\bibitem[Brenner and Subrahmanyam(1994)]{Brenner94}
Brenner, M. and Subrahmanyam, M.G., A simple approach to option valuation and hedging in the Black-Scholes~Model, Financ. Anal. J. {\bfseries 50}, 25--28 (1994)

\bibitem[Carr \& Madan(1999)]{Carr99}
Carr, P. and Madan, D., Option valuation using the Fast Fourier Transform, Journal of Computational Finance {\bfseries 2}, 61-73 (1999)

\bibitem[Carr \& al.(2002)]{Carr02} 
Carr, P., Geman, H., Madan, D., Yor, M., The Fine Structure of Asset Returns: An Empirical Investigation, Journal of Business {\bfseries 75(2)}, 305-332 (2002)

\bibitem[Carr and Wu(2003)]{Carr03}
Carr, P. and Wu, L., The Finite Moment Log Stable Process and Option Pricing, The Journal of Finance {\bfseries 58(2)}, 753-777 (2003)

\bibitem[Carr and Wu(2004)]{Carr04}
Carr, P. and Wu, L., Time-changed L\'evy processes and option pricing, Journal of Financial Economics {\bfseries 71}, 113-141 (2004)

\bibitem[Carr \& Wu(2012)]{Carr12}
Carr, P., Lee R. and Wu, L., Variance swaps on time-changed L\'{e}vy processes, Finance and Stochastics {\bfseries 16}, 335-355 (2012)

\bibitem[Cont \& Tankov(2004)]{Cont04}
Cont, R. and Tankov, P., Financial Modelling with Jump Processes, Chapman \& Hall, New York (2004)

\bibitem[Eberlein and Keller(1995)]{Eberlein95}
Eberlein, E. and Keller,U., Hyperbolic distributions in finance, Bernoulli {\bfseries 1(3)}, 281-299 (1995)

\bibitem[Eberlein(2001)]{Eberlein01}
Eberlein, E., Application of Generalized Hyperbolic L\'evy Motions to Finance. In: L\'evy Processes, Barndorff-Nielsen O.E., Resnick S.I., Mikosch T. (eds), Birkhauser, Boston, MA (2001)

\bibitem[Fang \& Osterlee(2008)]{Fang08}
Fang, F. and Oosterlee, C.W., A novel pricing method for European options based on Fourier cosine series expansions, SIAM Journal on Scientific Computing {\bfseries 31}, 826-848 (2008)

\bibitem[Figueroa-L\'opez {\it et al.}(2012)]{Figueroa12}
Figueroa-L\'opez, J.E., Lancette, S.R, , Lee, K. and Mi, Y., Estimation of NIG and VG models for high frequency financial data. In: Handbook of Modeling High-Frequency Data in Finance, F. Viens, M.C. Mariani, I. Florescu (eds.), John Wiley \& Sons, Hoboken, NJ (2012)

\bibitem[Flajolet {\it et al.}(1995)]{Flajolet95}
Flajolet, P., Gourdon, X. and Dumas, P., Mellin transforms and asymptotics: Harmonic sums, Theoretical Computer Science {\bfseries 144}, 3-58 (1995)

\bibitem[Fusai \& Meucci(2008)]{Fusai08}
Fusai, G. and Meucci, A., Pricing discretely monitored Asian options
under L\'evy processes, Journal of Banking \& Finance {\bfseries 32(10)}, 2076--2088 (2008)

\bibitem[Giannone \& al.(2008)]{Giannone08}
Giannone, D., Reichlin, L. and Small, D., Nowcasting: The real-time informational content of macroeconomic data, Journal of Monetary Economics {\bfseries 55(4)}, 665-676 (2008)

\bibitem[Glasserman(2004)]{Glasserman04}
Glasserman, P., Monte Carlo methods in financial engineering, Springer Science \& Business Media Vol.53, New York (2004)

\bibitem[Hanssen \& \O{}ig\aa{}rd(2001)]{Hanssen01}
Hanssen, A. and \O{}ig\aa{}rd, T.A., The Normal inverse Gaussian distribution: a versatile model for heavy-tailed stochastic processes, Proceedings - ICASSP, IEEE International Conference on Acoustics, Speech and Signal Processing {\bfseries 6}, 3986-3988 (2001)

\bibitem[Heston(1993)]{Heston93}
Heston, S., A Closed-Form Solution for Options with Stochastic Volatility with Applications to Bond and Currency Options, The Review of Financial Studies {\bfseries 6(2)}, 327-343 (1993)

\bibitem[Ivanov(2013)]{Ivanov13}
Ivanov, R.V., Closed Form Pricing of European Options for a Family of Normal Inverse Gaussian Processes, Journal of Stochastic Models {\bfseries 29(4)}, 435-450 (2013)

\bibitem[Keller-Ressel(2008)]{Keller08}
Keller-Ressel, M., Moment explosions and long-term behavior of affine stochastic volatility
models, arXiv:0802.1823 (2008)

\bibitem[Kirkby(2015)]{Kirkby15}
Kirkby, J. L., Efficient Option Pricing by Frame Duality with the Fast Fourier Transform, SIAM Journal on Financial Mathematics {\bfseries 6(1)}, 713-747 (2015)

\bibitem[Lewis(2001)]{Lewis01}
Lewis, A.L., A simple option formula for general jump-diffusion and other exponential L\'evy processes, Available at SSRN: https://ssrn.com/abstract=282110 (2001)

\bibitem[Luciano \& Semeraro(2010)]{Luciano10}
Luciano, E. and Semeraro, P., Multivariate time changes for L\'evy asset models: Characterization and calibration, Journal of Computational and Applied Mathematics {\bfseries 223(8)}, 1937-1953 (2010)

\bibitem[Madan {\it et al.}(1998)]{Madan98}
Madan, D., Carr, P. and Chang, E., The Variance Gamma Process and Option Pricing, European Finance Review {\bfseries 2}, 79-105 (1998)

\bibitem[Mandelbrot(1963)]{Mandelbrot63}
Mandelbrot, B., The Variation of Certain Speculative Prices, The Journal of Business {\bfseries 36(4)}, 384-419 (1963)

\bibitem[Matsuda(2006)]{Matsuda06}
Matsuda, K., Calibration of L\'evy Option Pricing Models: Applications to S\& P 500 Futures option, PhD Thesis City University of New York (2006)

\bibitem[Mechkov(2015)]{Mechkov15}
Mechkov, S., Fast-Reversion Limit of the Heston Model, Available at SSRN: https://ssrn.com/abstract=2418631 (2015)

\bibitem[Mittnik \& Rachev(2000)]{Mittnik00}
Mittnik, S. and Rachev, S., Stable Paretian models in finance, John Wiley \& Sons, Hoboken, NJ (2000)

\bibitem[Neuberger(1994)]{Neuberger94}
Neuberger, A., The log contract, Journal of Portfolio Management {\bfseries 20}, 74-80 (1994)

\bibitem[Papantolen(2008)]{Papantoleon08}
Papantoleon, A., An introduction to L\'evy Processes with applications in finance, arXiv:0804.0482 (2008)

\bibitem[Prause(1999)]{Prause99}
Prause, K, The generalized hyperbolic model:
estimation, financial derivatives and risk measures, PhD thesis,
Institut f\"{u}r Mathematische Statistik, Albert-Ludwigs-Universit\"{a}t
Freiburg (1999)


\bibitem[Rachev et al.(2011)]{Rachev11}
Rachev, S., Kim, Y., Bianchi, M., Fabozzi, F., Financial models with L\'evy processes and volatility clustering, John Wiley \& Sons, Hoboken, NJ (2011)

\bibitem[Ribeiro and Webber(2003)]{Ribeiro03}
Ribeiro, C. and Webber, N., A Monte Carlo Method for the Normal Inverse Gaussian Option Valuation Model using an
Inverse Gaussian Bridge, City University preprint (2003)

\bibitem[Rydberg(1997)]{Rydberg97}
Rydberg, T., The Normal inverse Gaussian L\'evy process: simulation and approximation, Communications in Statistics. Stochastic Models {\bfseries 13}, 887-910 (1997)

\bibitem[Saeb{\o}(2009)]{Saebo09}
Saeb{\o}, K., Pricing Exotic Options with the Normal Inverse Gaussian Market Model using Numerical Path Integration, Master's Thesis Norwegian University of Science and Technology (2009)

\bibitem[Schoutens(20003)]{Schoutens03}
Schoutens, W., L\'evy processes in finance: pricing financial derivatives, John Wiley \& Sons, Hoboken, NJ (2003)

\bibitem[Schoutens \& al.(2004)]{Schoutens04}
Schoutens, W., Simons, E., and Tistaert, J., A perfect calibration! Now what?, Wilmott magazine {\bfseries 2004(2)} (2004)

\bibitem[Su \& Fu(2000)]{Su00}
Su, Y., and Fu., M.C., Importance sampling in derivative securities pricing, 2000 Winter Simulation Conference Proceedings Vol. 1. (2000)

\bibitem[Taleb(2010)]{Taleb10}
Taleb, N.N., The Black Swan: The Impact of the Highly Improbable, Random House Publishing Group, New York (2010)

\bibitem[Tankov(2010)]{Tankov10}
Tankov, P., Pricing and Hedging in Exponential L\'evy Models: Review of Recent Results. In: Paris-Princeton Lectures on Mathematical Finance. Lecture Notes in Mathematics, vol 2003, Springer, Berlin, Heidelberg (2010)

\bibitem[Venter \& de Jongh(2002)]{Venter02}
Venter, J. and de Jongh, P., Risk estimation using the Normal inverse Gaussian distribution, The Journal of Risks {\bfseries 2}, 1-25 (2002)

\bibitem[Wilmott(2006)]{Wilmott06}
Wilmott, P., Paul Wilmott on Quantitative Finance, Wiley \& Sons, Hoboken, NJ, 2006

\bibitem[Zeng and Kwok(2014)]{Zeng14}
Zeng, P. and Kwok, Y.K., Pricing barrier and Bermudan style options under time-changed L\'evy processes: fast Hilbert transform approach, SIAM Journal on Scientific Computing {\bfseries 36(3)}, B450-B485 (2014)


\end{thebibliography}
\end{document}